\documentclass[12pt]{article}
\usepackage{amssymb}
\usepackage{amsmath}
\usepackage{amsfonts}
\usepackage{natbib}
\usepackage{setspace}
\usepackage{geometry}
\usepackage{graphicx}
\usepackage{hyperref}
\usepackage[title]{appendix}
 \hypersetup{
           breaklinks=true,   % splits links across lines
           colorlinks=false,   % displays links as colored text instead of blocks
           pdfusetitle=true,  % \title and \author values into pdf metadata
        }
\DeclareMathOperator*{\argmax}{argmax}

\newtheorem{theorem}{Theorem}

\newtheorem{corollary}{Corollary}

\newtheorem{example}{Example}

\newtheorem{lemma}{Lemma}

\newenvironment{proof}[1][Proof]{\noindent\textbf{#1.} }{\ \rule{0.5em}{0.5em}}
\def \indt{\mathbin{\vbox{\baselineskip=0pt\lineskip=0pt
  \moveright2.5pt\hbox{$\|$}
  \hrule height 0.2pt width 10pt}}}

\input{tcilatex}
\begin{document}

\title{Robust Analysis of Short Panels\thanks{This is a revision of the October 2023 CeMMAP working paper CWP20/23,  circulated under the title ``Identification analysis in models with unrestricted latent variables: Fixed effects and initial conditions." We are grateful for comments received at seminars at UCL, Duke, and the Montreal Econometrics Seminar, and at presentations at the 2023 Latin American Meeting of the Econometric Society in Bogota and the $34^{\text{th}}$ EC$^2$ Meeting at the University of Manchester in December 2023. Financial support from the UK Economic and Social Research Council through a grant (RES-589-28-0001) to the ESRC Centre for Microdata Methods and Practice (CeMMAP) is gratefully acknowledged.}}
\author{Andrew Chesher \\ UCL and CeMMAP\thanks{Corresponding author. Address: Department of Economics, University College London, Gower Street, London WC1E 6BT, United Kingdom. Email: andrew.chesher@ucl.ac.uk.} \and Adam M. Rosen \\ Duke University and CeMMAP\thanks{Address: Adam Rosen, Department of Economics, Duke University, 213 Social Sciences Box 90097, Durham, NC 27708; Email: adam.rosen@duke.edu.} \and Yuanqi Zhang \\ UCL and CeMMAP\thanks{Address: Department of Economics, University College London, Gower Street, London WC1E 6BT, United Kingdom. Email: uctpyqz@ucl.ac.uk.}}
\maketitle

\begin{abstract}
Many structural econometric models include latent variables on whose
probability distributions one may wish to place minimal restrictions.
Leading examples in panel data models are individual-specific variables
sometimes treated as \textquotedblleft fixed effects\textquotedblright\ and,
in dynamic models, initial conditions. This paper presents a generally
applicable method for characterizing sharp identified sets when models place
no restrictions on the probability distribution of certain latent variables
and no restrictions on their covariation with other variables. In our analysis latent variables on which restrictions are undesirable are removed, leading to econometric analysis robust to misspecification of restrictions on their distributions which are commonplace in the applied panel data literature.  Endogenous explanatory variables are easily accommodated. Examples of application to some static and dynamic binary, ordered and multiple discrete choice and censored panel data models are presented.
\end{abstract}

\section{Introduction}

This paper deals with models of processes delivering values of outcomes, $Y$%
, given values of exogenous variables, $Z$, and latent, that is unobserved,
variables $U$ and $V$. The models that are the focus of this paper all leave
the distribution of $V$ on its known support and its covariation with all
other variables completely unrestricted. By contrast, latent variable $U$
may be required to be, to some degree, independent of $Z$.

Leading examples of latent variables in structural econometric models
employed in practice on whose distribution one may not want to impose
restrictions are the individual-specific unobserved variables included in
many panel data models, sometimes called \textquotedblleft fixed
effects\textquotedblright\ and the historic values of outcomes dynamically
determined by a process, commonly called \textquotedblleft initial
conditions\textquotedblright .

The following example has both elements, a \textquotedblleft fixed
effect\textquotedblright , $C$, and an initial condition, $Y_{10}$.

\begin{example}
\label{ex 1}A dynamic binary response model specifies that for all $t\in
[T]\equiv\{1,...,T\}$,
\begin{equation}  \label{example eq}
Y_{1t}=\left\{ 
\begin{array}{lll}
1 & , & \alpha Y_{2t}+Z_{t}\beta +\gamma Y_{1t-1}+C+U_{t} \geq 0\text{,} \\ 
0 & , & \alpha Y_{2t}+Z_{t}\beta +\gamma Y_{1t-1}+C+U_{t} \leq 0\text{,}%
\end{array}%
\right.
\end{equation}%
with $Y_{1t}=0$ or $Y_{1t}=1$ permitted when both inequalities hold,%
\footnote{%
This is equivalent to the representation $Y_{1t}=1[\alpha Y_{2t}+Z_{t}\beta
+\gamma Y_{1t-1}+C+U_{t}>0]$ when the indicator function $1\left[a > b %
\right]$ takes the value $1$ if $a > b$, $0$ if $a < b$, and either value if 
$a=b$.} and $U 
%TCIMACRO{\TeXButton{indt}{\indt}}%
%BeginExpansion
\indt%
%EndExpansion
Z$ where%
\begin{equation*}
U\equiv (U_{1},\dots ,U_{T})\text{,}\qquad Z\equiv (Z_{1},\dots ,Z_{T})\text{%
.}
\end{equation*}%
Realizations of $(Y,Z)$ are observed where 
\begin{equation*}
Y=(Y_{11},\dots ,Y_{1T},Y_{21},\dots ,Y_{2T}),
\end{equation*}%
and $Z_{t}$ and $Y_{2t}$ may be vectors. If the value of $Y_{10}$ is
observed then unrestricted $V=C$, otherwise $V=(C,Y_{10})$. If $\alpha $ is
not restricted equal to zero there are endogenous explanatory variables.
\end{example}

In many cases found in practice in which $T$ is large, the value $V$ takes for each observational unit is identified. In this case econometric analysis can proceed placing no restrictions at all
on the distribution of $V$ and treating it as a parameter to be estimated. When $T$ is not large this is unattractive because there may be intolerable inaccuracy in the estimation of $V$ which may contaminate estimates of other parameters. Additionally
there is the incidental parameters problem set out in \cite{Neyman/Scott:48} and reviewed in \cite{Lancaster:00}.

Faced with this problem, for small $T$, most papers proceed to obtain information on
structural features by placing distributional restrictions on $V$. Section 
\ref{Section: literature} lists many examples. It is good to know
what knowledge of structural features can be obtained absent such
restrictions. That allows the force of distributional restrictions on $V$ to
be assessed and offers the possibility of detecting misspecification. This
paper shows how that knowledge can be obtained. The results given here open
the way to a relatively robust analysis of models like panel models with
fixed effects in which there are latent variables on which one desires to
place no distributional restrictions. 

This paper presents characterizations of identified sets of structures and
structural features in models admitting unobserved variables such as $V$
whose distribution is unrestricted. There can be endogenous explanatory
variables as in Example \ref{ex 1} when $\alpha \neq 0$. A model may be 
\emph{incomplete} in the sense that, given values of all observed and all
unobserved variables and a specification of parameter values and functional
forms, the model can deliver a nonsingleton set of values of outcomes.
Identified sets are characterized by systems of moment inequalities.
Estimation and inference can proceed using established econometric methods.

The strategy employed here removes unrestricted latent variables, $V$,
by projection.\footnote{Our eschewal of restrictions on the distribution of $V$ accords with the approach in \cite{Neyman/Scott:48} in which the elements of $V$ are treated as parameters, subject to no restrictions.} We derive, for each value of the observed variables, the set of values of unobserved $U$ compatible with that value. Values of $U$ in such a set are associated with alternative values of $V$. Typically a value of $U$ in such a set can deliver more than one value of $Y$.  So, on removing latent variables $V$, there remains an incomplete model.\footnote{If the model is incomplete before projection then different values of $V$ can deliver different sets of values of $Y$.} 

Identification analysis is conducted in the context of the Generalized Instrumental Variable (GIV) framework introduced in \cite{chesher2017generalized} in which probability distributions of such sets of values of $U$ induced by the observed distributions of outcomes are essential elements.

Section \ref{Section: literature} considers the relationship of this work to some other results in the literature. Section \ref{Section Identified sets} presents characterizations of identified sets of structures. Sections \ref{Section: linear panel model} to \ref{Section: Tobit panels} set out applications to linear panel models and to models of binary response
panels, ordered choice panels, multiple discrete choice panels, simultaneous binary outcome panels, and models of panels with censored continuous outcomes.

\section{Related literature\label{Section: literature}}

\cite{rasch1960probabilistic}, \cite{rasch1961general}, \cite%
{andersen1970asymptotic}, and \cite{chamberlain2010binary} study point
identifying static panel models (i.e. $\gamma =0$ in (\ref{example eq}))
with restrictions requiring $U_{1},...,U_{T}$ to be independent over time
and distributed independently of $Z$ and independently of the fixed effect
and each with logistic marginal distributions. Like all the papers referred
to in this section, except one paper which is noted, these models do not
admit endogenous explanatory variables.

In the linear panel data model with fixed effects, differencing across time
periods removes the fixed effect, delivering events whose probability of
occurrence can be known and is invariant with respect to changes in the
value of the fixed effect. Under suitable support restrictions this leads to
point identification. In nonlinear panel data models with fixed effects, this simple differencing strategy does not apply.  Nonetheless, in the Rasch-Andersen-Chamberlain set up, events whose probabilities of occurrence are invariant to changes in the value of the fixed effect are found. Under particular distributional
restrictions point identification results. More recent papers on nonlinear panel data models have taken a similar approach.

\cite{honore2000panel} study a dynamic model as in (\ref{example eq}) but
with no endogenous explanatory variable ($\alpha =0$) with the $U_{t}$'s
independent of the fixed effect, independent over time, distributed
independently of $Z$ and with logistic distributions. That paper also
studies a case in which the logistic distribution restriction is dropped and
a case with multinomial logit panels with latent variables $U$ independent
of the fixed effects and independent of $Z$. \cite{honore2019panel} extends
this work, studying multivariate dynamic panel data logit models with fixed
effects. Many papers, like these, invoke restrictions requiring
independence between $U_{t}$'s and the fixed effect conditional on some of
the other observable variables including \cite{honore2006bounds}, \cite{honore2021identification},\footnote{This paper considers models in which there are simultaneous equations in binary outcomes and so, endogenous explanatory variables.} \cite{Dobronyi/Gu/Kim:21}, \cite{honore2022dynamic}, \cite{Davezies/D'Haultfoeuille/Laage:22}, \cite{Kitazawa:22}, \cite{Bonhomme/Dano/Graham:23}, \cite{Dano:23}, \cite%
{Davezies/D'Haultfoeuille/Mugnier:23}, and \cite{honore2021dynamic}. Such
independence restrictions are not imposed here.\footnote{%
One approach in such settings, demonstrated by e.g. \cite{honore2022dynamic}
and \cite{honore2021dynamic}, is the functional differencing approach
developed in \cite{Bonhomme:12}. This however requires knowledge of the
distribution of $F_{Y|Z,C}$, which one does not have in models such as that
of Example 1 without knowledge of the joint distribution of $U$ and $C$.} This permits for example the $U_t$'s to exhibit heteroskedastic variation with observational-unit-specific fixed effects.

There are many papers studying panel models of binary outcomes and multiple
discrete choice under conditional \textit{stationarity} restrictions on the
distribution of the time varying latent variables introduced in \cite%
{manski1987semiparametric}. These papers include \cite%
{Chernozhukov/Fernandez-Val/Hahn/Newey:09}, \cite{shi2018estimating}, \cite%
{Gao/Li:20}, \cite{khan2021inference}, \cite{pakes2021unobserved}, \cite%
{pakes2022moment}, \cite{Dobronyi/Ouyang/Yang:23}, \cite{khan2023identification}, and \cite%
{mbakop2023identification}.

In all of these cases the stationarity restriction placed on time-varying
unobservable heterogeneity is required to hold conditional on the value of
the fixed effect and the observable exogenous variables, which restricts the
covariation of the fixed effect and $U$.\footnote{%
In the binary response specification (\ref{example eq}) conditional stationarity implies
that for all $z$, $F_{U_1|Z=z,C=c} = F_{U_2|Z=z,C=c}$ and $%
F_{U_1|Z=z,C=c^{\prime }} = F_{U_2|Z=z,C=c^{\prime }}$ for any $c,c^{\prime
} $, which restricts how the conditional distribution of $U$ can change with
values of the fixed effect $C$. As pointed out by \cite%
{Chernozhukov/Fernandez-Val/Hahn/Newey:09} the stationarity restriction $%
U_t|C,Z \overset{d}{=}U_1|C,Z$ for all $t$ is equivalent to $(U_t,C)|Z 
\overset{d}{=}(U_1,C)|Z$ for all $t$.} In contrast, the models considered in
this paper impose no restrictions on the covariation of the fixed effect
with any variable.

The only previous paper of which we are aware that provides partial identification analysis for discrete outcome panel data models absent restrictions on the covariation of the fixed effect with any other variables is \cite{aristodemou2021semiparametric}. That paper provides set-identifying moment inequalities in panel data models of binary response and ordered choice when the covariation of the fixed effects with other variables is unrestricted.  The results developed in this paper provide a rule-directed procedure for enumerating all events whose probability is invariant with respect to the value of unrestricted latent variables thereby delivering sharp set identification for these and other nonlinear panel data models.

Application of sharp set identification analysis to panel data models with censored outcomes is demonstrated in Section \ref{Section: Tobit panels}. Observable implications in the form of moment equalities are derived for such models with Tobit-type censoring at zero in both static and dynamic contexts in Honor{\'e} (1992, 1993),\nocite{Honore:92}\nocite{Honore:93} \cite{Honore/Hu:2002}, and \cite{Hu:2002}, all in models in which the $U_t$'s satisfy the conditional stationarity assumption that has also been used in discrete outcome panel models. The only previous paper of which we are aware that provides identification analysis for censored outcome panel models without restricting the covariation of the fixed effect with other variables is \cite{Khan/Ponomareva/Tamer:16} (KPT), which provides the sharp identified set for slope coefficient $\beta$ in a static two-period model in which $U_2 - U_1$ and $Z$ are independent. Extensions are provided to some specialized dynamic models with two periods of observations with an observed initial condition and inequality restrictions on parameters.  The analysis here additionally accommodates more periods, unobserved initial conditions, and endogenous explanatory variables.  Like KPT we allow the censoring value to vary and to be endogenous, nesting the classical case of fixed censoring found in Tobit models.

%nonlinear panel data models absent restrictions on the covariation of the fixed effect with any other variables, namely \cite{aristodemou2021semiparametric} and \cite{Khan/Ponomareva/Tamer:16}. \cite{aristodemou2021semiparametric} provides set-identifying moment inequalities in panel models of binary response and static models of ordered choice. \cite{Khan/Ponomareva/Tamer:16} characterize .  Both papers consider the impact of independence restrictions between observable covariates and time-varying unobservable heterogeneity, without any distributional restrictions on the fixed effect.\footnote{\cite{Khan/Ponomareva/Tamer:16} additionally provide partial identification analysis for censored outcome panels that impose conditional stationarity restrictions.}  The results developed in this paper provide a rule-directed procedure for finding \emph{all} events whose probability is invariant with respect to the value of unrestricted latent variables such as fixed effects and initial conditions, and thereby deliver sharp set identification for these and other nonlinear panel data models.

This paper presents a generally applicable approach to identification
analysis in a wide class of nonlinear panel data models in which there are distributionally unrestricted latent variables and gives examples of the results it produces. Most of our examples feature discrete outcomes, but the application to the censored outcome model in Section \ref{Section: Tobit panels} demonstrates that the analysis applies more broadly.

\section{\label{Section Identified sets}Identified sets}

First the notation employed in this paper is introduced.\medskip

\noindent \textbf{Notation}. \textit{Generically} $\mathcal{R}_{A}$\textit{\
denotes the support of random variable }$A$\textit{\ and }$L_{A|Z=z}$\textit{%
\ denotes a conditional probability distribution of random variable }$A$%
\textit{\ given }$Z=z$\textit{. }$L_{A|Z=z}(\mathcal{S})$\textit{\ is the
conditional probability }$A$\textit{\ takes a value in set }$\mathcal{S}$%
\textit{\ given }$Z=z$\textit{. }$\mathcal{L}_{A|Z}\equiv \{L_{A|A=z};z\in
R_{Z}\}$\textit{\ is the collection of conditional distributions delivered
by a joint distribution }$L_{AZ}$\textit{\ when the support of }$Z$\textit{\
is }$\mathcal{R}_{Z}$\textit{. }$A%
%TCIMACRO{\TeXButton{indt}{\indt}}%
%BeginExpansion
\indt%
%EndExpansion
B$ \textit{denotes }$A$\textit{\ and }$B$\textit{\ are independently
distributed.} \textit{Sets and set-valued random variables are expressed
using calligraphic font. Collections of sets are expressed using sans serif
font. }$%
%TCIMACRO{\U{211d} }%
%BeginExpansion
\mathbb{R}
%EndExpansion
$ \textit{denotes the real line. The empty set is denoted $\emptyset $. For $T > 1$, notation $[T]$ denotes $\{1,...T\}$. For any random vectors $X_1,...,X_T$ notation $\Delta_{ts} X \equiv X_t - X_s$ is used throughout. }
\medskip

Variables $Y$ are endogenous outcomes, variables $Z$ are exogenous\footnote{In the sense that their values are not affected by the evolution of the process.} and variables $U$ and $V$ are latent variables. Random vectors $(Y,Z,U,V)$ are defined on a probability space $(\Omega ,\mathsf{L},\mathbb{P})$ endowed with the Borel sets on $\Omega $. The support of $(Y,Z,U,V)$ is a subset of a finite dimensional Euclidean space. The sampling process identifies $F_{YZ}$, equivalently the collection of conditional distributions $\mathcal{F}_{Y|Z}$ and $F_Z$, as occurs for example under random sampling of observational units. It is assumed throughout for ease of exposition that each observational unit delivers the same number of observations, but unbalanced panels are easily accommodated with some added notation.

Models place restrictions on a structural function $h:\mathcal{R}%
_{YZUV}\rightarrow 
%TCIMACRO{\U{211d} }%
%BeginExpansion
\mathbb{R}
%EndExpansion
$ which specifies the combinations of these variables that can occur \textit{%
via} the following restriction.\footnote{%
In the case of (\ref{example eq}) a suitable $h$ function would be%
\begin{equation*}
h(Y,Z,U,V)=\dsum\limits_{t=1}^{T}\max \{0, (1 - 2 Y_{1t})\cdot(\alpha
Y_{2t}+Z_{t}\beta +\gamma Y_{1t-1}+C+U_{t})\} .
\end{equation*}%
with $V=(C,Y_{0})$.}%
\begin{equation*}
\mathbb{P}[h(Y,Z,U,V)=0]=1
\end{equation*}%
Models place restrictions on the conditional probability distributions of $U$
given $Z$ which are elements of a collection $\mathcal{G}_{U|Z}$. Coupled
pairs $(h,\mathcal{G}_{U|Z})$ are called \emph{structures}. A model $%
\mathcal{M}$ is a collection of structures that obey the restrictions
imposed a priori on the data generation process. This
paper provides sharp identification analysis of structures $(h,\mathcal{G}%
_{U|Z}) \in \mathcal{M}$ and functionals thereof given knowledge of $%
\mathcal{F}_{Y|Z}$.

The essential element of the models considered here is that they place \emph{%
no} restrictions on the marginal distribution of $V$ and \emph{no}
restrictions on the covariation of $V$ with $(Z,U)$.

This paper shows how the framework set out in \cite{chesher2017generalized}
(CR) can be used to study cases with unobserved variables whose distribution
and covariation with other variables is not subject to restrictions. The
support of any initial condition components of $V$ is assumed known and the
support of all ``fixed effect'' components of $V$ is assumed to be the
entirety of the Euclidean space in which it resides. It is straightforward
to generalize the analysis to cases in which the support of the fixed effect
is restricted.

For all characterizations of identified sets of values of the pair $(h,\mathcal{G}_{U|Z}) \in \mathcal{M}$ it is assumed that a priori restrictions on $\mathcal{G}_{U|Z}$ are such that $U|Z$ is restricted absolutely continuous with respect to Lebesgue measure almost surely. This renders the boundary of sets $\mathcal{U}^{\ast}(y,z;h)$ to be measure zero with respect to any distribution $G_{U|Z=z}$. It is convenient to define the structural function $h$ such that sets $\mathcal{U}^{\ast }(Y,Z;h)$ are closed almost surely in the usual Euclidean topology,
and we do so here, but this is of no substantive consequence and can be
relaxed.\footnote{With some care equivalent results could be obtained allowing for random open
sets and random closed sets, or by working with an alternative topology in
which the sets under consideration are closed, such as the discrete topology
when $\mathcal{R}_Y$ is discrete. One could also allow sets of values of
unobservables that deliver ``ties'' in the optimal choice of discrete
outcome with positive probability, and apply results of CR, with suitable
care.}

Taken together the restrictions set out above ensure that Restrictions A1 -
A6 of CR hold in the models considered, suitably modified to accommodate
unobservable variables $(U,V)$ with the distribution of $V$ unrestricted.%
\footnote{The latent variables $U$ in restrictions A1-A6 of CR should be taken to
include both the variables $U$ and $V$ of this paper. For completeness,
these restrictions, adapted to the present context, are collected in
Appendix \ref{Appendix: CR restrictions}.}

Theorem \ref{Outerset th 1} provides a characterization of the identified
set of structures, denoted $\mathcal{I}(\mathcal{M},\mathcal{F}_{Y|Z})$,
delivered by a model $\mathcal{M}$ and a collection of distributions, $%
\mathcal{F}_{Y|Z}$. This is the collection of distributions \emph{marginal}
with respect to $V$ obtained from some collection $\mathcal{F}_{YV|Z}$.

\begin{theorem}
\label{Outerset th 1}Let $\mathcal{R}_{V}$ denote the support of $V$. Define 
$\mathcal{U}^{\ast }(y,z;h)$ as follows. 
\begin{equation}  \label{Ustar set}
\mathcal{U}^{\ast }(y,z;h)\equiv \{u:\exists v\in \mathcal{R}_{V}\quad \text{%
such that}\quad h(y,z,u,v)=0\}
\end{equation}%
Let $\mathcal{F}_{Y|Z}$ be a collection of distributions whose members are
marginal distributions of the members of some collection of distributions $%
\mathcal{F}_{YV|Z}$. The set of structures $(h,\mathcal{G}_{U|Z})$
identified by model $\mathcal{M}$ and the collection of distributions $%
\mathcal{F}_{Y|Z}$ comprises all structures admitted by the model $\mathcal{M%
}$ such that for all $z\in \mathcal{R}_{Z}$, the probability distribution $%
G_{U|Z=z}\in \mathcal{G}_{U|Z}$ is selectionable with respect to the
conditional distribution of the random set $\mathcal{U}^{\ast }(Y,Z;h)$
delivered by the probability distribution $F_{Y|Z=z}\in \mathcal{F}_{Y|Z}$%
.\medskip
\end{theorem}

Formally the Theorem defines the identified set of structures $(h,\mathcal{G}%
_{U|Z})$ as%
\begin{multline}
\mathcal{I}(\mathcal{M},\mathcal{F}_{Y|Z})\equiv \left\{ (h,\mathcal{G}%
_{U|Z})\in \mathcal{M}:G_{U|Z=z}\preceq \mathcal{U}^{\ast }(Y,Z;h)\right. \\
\left. \text{conditional on }Z=z\quad \text{a.e.}\quad z\in \mathcal{R}%
_{Z}\right\}\text{,}\label{Ustar selectionability}
\end{multline}
where, as in \cite{chesher2020generalized}, for any random variable $A$ with
distribution $F_{A}$ and random set $\mathcal{A}$, $F_{A}\preceq \mathcal{A}$
denotes that $F_{A}$ is selectionable with respect to the distribution of $%
\mathcal{A}$.\footnote{%
The probability distribution of random variable $A$ is selectionable with
respect to the probabilty distribution of random set $\mathcal{A}$ when
there exists (i) $\tilde{A}$ having the same distribution as $A$, and (ii) $%
\widetilde{\mathcal{A}}$ having the same distribution as $\mathcal{A}$, both
defined on the same probability space such that $\mathbb{P}[\tilde{A}\in 
\widetilde{\mathcal{A}}]=1$. See Definition 2 of Chesher and Rosen (2020).}%
\bigskip

The proof relies on the following Lemma.

\begin{lemma}
Define
\begin{equation}
\mathcal{Y}^{\ast }(z,u;h)\equiv \{y:\exists v\in \mathcal{R}_{V}\quad \text{%
such that}\quad h(y,z,u,v)=0\}\text{.}  \label{ystar set definition}
\end{equation}%
The sets $\mathcal{Y}^{\ast }(z.u;h)$ and $\mathcal{U}^{\ast }(y,z;h)$
possess the duality property%
\begin{equation*}
\forall z,y^{+},u^{+}\quad y^{+}\in \mathcal{Y}^{\ast }(z,u^{+};h)\iff
u^{+}\in \mathcal{U}^{\ast }(y^{+},z;h)\text{.}
\end{equation*}
\end{lemma}

\begin{proof}
The result follows because%
\begin{equation*}
y^{+}\in \mathcal{Y}^{\ast }(z,u^{+};h)\iff \exists v\in \mathcal{R}%
_{V}\quad \text{such that}\quad h(y^{+},z,u^{+},v)=0
\end{equation*}%
\begin{equation*}
u^{+}\in \mathcal{U}^{\ast }(y^{+},z;h)\iff \exists v\in \mathcal{R}%
_{V}\quad \text{such that}\quad h(y^{+},z,u^{+},v)=0.
\end{equation*}
\end{proof}

The proof of Theorem \ref{Outerset th 1} above proceeds as the proof of
Theorem 2 in CR, replacing $U$ sets with $U^{\ast }$ sets.

The identified set of structures can be characterized as shown in Corollary %
\ref{Corollary: Artstein} using the characterization of selectionability
given in \cite{artstein1983distributions}, as in Corollary 2 of CR.

\begin{corollary}
\label{Corollary: Artstein}Let $\mathsf{F}(\mathcal{R}_{U})$ denote the
collection of closed sets on the support of $U$. The set of structures
identified by model $\mathcal{M}$ and the collection of distributions $\mathcal{F}_{Y|Z}$ is as follows. 
\begin{multline}  \label{Corollary 1 sharp set}
\mathcal{I}(\mathcal{M},\mathcal{F}_{Y|Z})\equiv \left\{ (h,\mathcal{G}%
_{U|Z})\in \mathcal{M}:\forall \mathcal{S}\in \mathsf{F}(\mathcal{R}%
_{U})\right. \\
\left. F_{Y|Z=z}(\{y:\mathcal{U}^{\ast }(y,z;h)\subseteq \mathcal{S}\})\leq
G_{U|Z=z}(\mathcal{S})\text{ a.e. }z\in \mathcal{R}_{Z}\right\} \text{.}
\end{multline}
\end{corollary}

\subsection*{\label{Section: Remarks}Remarks}

\begin{enumerate}
\item The probability $F_{Y|Z=z}(\{y:\mathcal{U}^{\ast }(y,z;h)\subseteq 
\mathcal{S}\})$ is the probability conditional on $Z=z$ of the occurrence of
a value of $Y$ that can \emph{only} occur when $U\in \mathcal{S}$. We will
refer to such a probability as a \textit{containment probability} and employ
the notation $\mathcal{A}(\mathcal{S},z;h)\equiv \{y:\mathcal{U}^{\ast
}(y,z;h)\subseteq \mathcal{S}\}$.

\item Because the inequalities defining $\mathcal{I}(\mathcal{M},\mathcal{F}%
_{Y|Z})$ only involve probabilities of events under which $U^{\ast }$ sets
are subsets of test sets, $\mathcal{S}$, the collection of test sets $%
\mathsf{F}(\mathcal{R}_{U})$ in the definition of $\mathcal{I}(\mathcal{M},%
\mathcal{F}_{Y|Z})$ can, for each $z\in \mathcal{R}_{Z}$, be replaced by the
collection of all unions of $U^{\ast }$ sets, 
\begin{equation}
\mathsf{U}^{\ast }(z;h)\equiv \left\{ \dbigcup\limits_{y\in \mathcal{Y}}%
\mathcal{U}^{\ast }(y,z;h):\mathcal{Y}\subseteq \mathcal{R}_{Y}\right\} .
\label{all unions definition}
\end{equation}

\item Let $\mathsf{Q}(z;h)$ be a core determining collection (CDC) of sets
such that if the inequality%
\begin{equation*}
F_{Y|Z=z}(\{y:\mathcal{U}^{\ast }(y,z;h)\subseteq \mathcal{S}\})\leq
G_{U|Z=z}(\mathcal{S})
\end{equation*}%
holds for all $\mathcal{S}\in \mathsf{Q}(z;h)$ then it holds for all $%
\mathcal{S}\in \mathsf{F}(\mathcal{R}_{U})$. The collection $\mathsf{U}%
^{\ast }(z;h)$ defined in (\ref{all unions definition}) is such a CDC for the specified value of $z$.

\begin{enumerate}
\item If disjoint $\mathcal{S}_{1}$ and $\mathcal{S}_{2}$ are members of a
CDC and $\mathcal{A}(\mathcal{S}_{1},z;h)\cap \mathcal{A}(\mathcal{S}%
_{2},z;h)=\emptyset $, which occurs for example when all $U^{\ast }$ sets
are connected sets, then $\mathcal{S}=\mathcal{S}_{1}\cup \mathcal{S}_{2}$
can be excluded from the CDC. Theorem 3 of CR applies and gives further
refinements.

\item If $\mathcal{S}_{1}$ and $\mathcal{S}_{2}$ are members of a CDC with $%
\mathcal{S}\equiv \mathcal{S}_{1}\cap \mathcal{S}_{2}$ and $\mathcal{S}%
_{1}\cup \mathcal{S}_{2}=\mathcal{R}_{U}$, and 
\begin{equation*}
\mathcal{A}(\mathcal{S}_{1},z;h)\cup \mathcal{A}(\mathcal{S}_{2},z;h)=%
\mathcal{R}_{Y}
\end{equation*}%
\begin{equation*}
\mathcal{A}(\mathcal{S}_{1},z;h)\cap \mathcal{A}(\mathcal{S}_{2},z;h)=%
\mathcal{A}(\mathcal{S},z;h)
\end{equation*}%
then $\mathcal{S}$ can be excluded from the CDC by results in \cite%
{Luo/Wang:16} and \cite{Ponomarev:22}.

\item In particular applications some members of a CDC need not be
considered because they deliver inequalities that are dominated by others.
\end{enumerate}

\item If there is additionally the restriction $U%
%TCIMACRO{\TeXButton{indt}{\indt}}%
%BeginExpansion
\indt%
%EndExpansion
Z$ then $\mathcal{G}_{U|Z}=\{G_{U}\}$ and there is the following
simplification.%
\begin{multline}
\mathcal{I}(\mathcal{M},\mathcal{F}_{Y|Z})\equiv \left\{ 
\begin{array}{c}
\mathstrut \\ 
\mathstrut \end{array}
(h,{G}_{U})\in \mathcal{M}:\quad \forall \mathcal{S}\in \mathsf{F}(\mathcal{R}_{U})\right. \\ \left. \sup_{z\in \mathcal{R}_{Z}}F_{Y|Z=z}(\{y:\mathcal{U}%
^{\ast }(y,z;h)\subseteq \mathcal{S}\})\leq G_{U}(\mathcal{S}) \begin{array}{c}
\mathstrut \\ 
\mathstrut \end{array} \right\}\text{.}
\end{multline}
Quantile and mean independence restrictions can be accommodated.

\item Sets $\mathcal{U}^{\ast}(y,z;h)$ compatible with realizations $(y,z)$ of observable variables can be obtained in a variety of ways. In many models $h(y,z,u,v)=0$ if and only if there exist functions $d_{t}(\cdot;h):\mathcal{R}_{YZUV} \rightarrow \mathbb{R}$ such that for all $t\in [T]$, $d_{t}(y,z,u,v)=0$. The generalized inverse of this mapping with respect to $v$ is
\begin{equation*}
\mathcal{D}_{t}(y,z,u;h)\equiv \{v:d_{t}(y,z,u,v)=0\}\text{,}
\end{equation*}%
point-valued when $d_{t}(y,z,u,v)$
is strictly monotone in scalar $v$ and more generally set-valued.

The $U^{\ast }$ sets can then be written as
\begin{equation*}
\mathcal{U}^{\ast }(y,z;h)=\left\{ u:\dbigcap\limits_{t\in [T]}
\mathcal{D}_{t}(y,z,u;h)\neq \emptyset \right\}\text{.}
\end{equation*}
In models in which outcomes $Y_{1t}$ are determined by a weakly monotone transformation of an index function that is linear in $v$, sets $\mathcal{D}_{t}(y,z,u)$ are characterized by linear equalities and inequalities.  Variables $v$ can then be analytically removed from these linear systems, for example by way of Fourier-Motzkin elimination, to obtain linear inequalities characterizing $\mathcal{U}^{\ast }(y,z;h)$.

This can be used both in the examples studied in this paper, in which scalar fixed effects enter additively in an index function, as well as in more general models that allow individual-specific coefficients in such index functions.

\item Many of our illustrative examples will employ the restriction that $U$
and $Z$ are fully independent, but the characterizations afforded by Theorem %
\ref{Outerset th 1} and Corollary \ref{Corollary: Artstein} allow for a much
wider variety of restrictions on the collection of conditional distributions $\mathcal{G}_{U|Z}$. For example, restrictions
could require that $U_{t}\mathbin{\vbox{\baselineskip=0pt\lineskip=0pt
  \moveright2.5pt\hbox{$\|$}
  \hrule height 0.2pt width 10pt}}(Z_{1},...,Z_{t})$ for all $t$, while
permitting dependence between $U_{t}$ and $Z_{s}$ for $s>t$, hence allowing
models that impose only weak exogeneity.

\item Identified sets of values of a structural feature, defined as a
functional, $\theta \left( (h,\mathcal{G}_{U|Z})\right) $, are obtained by
projection.
\begin{equation*}
\mathcal{I}_{\theta }(\mathcal{M},\mathcal{F}_{Y|Z})=\{\theta (\left( h,%
\mathcal{G}_{U|Z}\right) ):\left( h,\mathcal{G}_{U|Z}\right) \in \mathcal{I}(%
\mathcal{M},\mathcal{F}_{Y|Z})\}.
\end{equation*}%
An example of such a structural feature is a vector of coefficients
multiplying included exogenous variables in models in which $h$ is
parametrically specified with a linear index restriction.

\item \label{NumberedList: capcon}Outer sets for the projection of the
identified set of structures onto the space of structural functions can be
obtained. Impose the restriction $U%
%TCIMACRO{\TeXButton{indt}{\indt}}%
%BeginExpansion
\indt%
%EndExpansion
Z$ and let there be no further restrictions on $G_{U}$. All structures in $%
\mathcal{I}(\mathcal{M},\mathcal{F}_{Y|Z})$ satisfy the inequality%
\begin{equation*}
\sup_{z\in \mathcal{R}_{Z}}F_{Y|Z=z}(\{y:\mathcal{U}^{\ast }(y,z;h)\subseteq 
\mathcal{S}\})\leq G_{U}(\mathcal{S})
\end{equation*}%
and applying this with $\mathcal{S}$ replaced by its complement delivers 
\begin{equation*}
G_{U}(\mathcal{S})\leq \inf_{z\in \mathcal{R}_{Z}}F_{Y|Z=z}(\{y:\mathcal{U}%
^{\ast }(y,z;h)\cap \mathcal{S}\neq \emptyset \}).
\end{equation*}%
Let $\mathcal{H}(\mathcal{M})$ denote the set of structural functions
admitted by model $\mathcal{M}$. There is the following outer identified set
on the space of structural functions. 
\begin{multline}
\mathcal{I}_{h}(\mathcal{M},\mathcal{F}_{Y|Z})\equiv \left\{ 
\begin{array}{c}
\mathstrut \\ 
\mathstrut%
\end{array}%
h\in \mathcal{H}(\mathcal{M}):\quad \forall \mathcal{S}\in \mathsf{F}(\mathcal{R}_{U})\right. \\
\sup_{z\in \mathcal{R}_{Z}}F_{Y|Z=z}(\{y:\mathcal{U}^{\ast }(y,z;h)\subseteq 
\mathcal{S}\})\leq \\
\left. \inf_{z\in \mathcal{R}_{Z}}F_{Y|Z=z}(\{y:\mathcal{U}^{\ast
}(y,z;h)\cap \mathcal{S}\neq \emptyset \})\right\}  \label{outer eq}
\end{multline}

\item There is an alternative characterization of the identified set of
structures
\begin{multline}
\mathcal{I}(\mathcal{M},\mathcal{F}_{Y|Z})\equiv \left\{ (h,\mathcal{G}_{U|Z})\in \mathcal{M}:F_{Y|Z=z}\preceq \mathcal{Y}^{\ast }(z,U;h)\right. \\
\left. \text{conditional on }Z=z\quad \text{a.e.}\quad z\in \mathcal{R}_{Z},\right\}\label{Y selectionability}
\end{multline}
where the set $\mathcal{Y}^{\ast }(\cdot ,\cdot ;\cdot )$ is defined in (\ref%
{ystar set definition}).\footnote{See for example \cite{Beresteanu/Molchanov/Molinari:09} and \cite{Molinari:Handbook} and further references therein for such characterizations. Theorem 1 of CR implies equivalence of characterizations of the form (\ref{Ustar selectionability}) and (\ref{Y selectionability}).} Using the Artstein characterization of
selectionability this leads to the following representation
\begin{multline}
\mathcal{I}(\mathcal{M},\mathcal{F}_{Y|Z})\equiv \left\{ (h,\mathcal{G}%
_{U|Z})\in \mathcal{M}:\forall \mathcal{K}\in \mathsf{\ K}(\mathcal{R}%
_{Y})\right. \\
\left. G_{U|Z=z}(\left\{ u:\mathcal{Y}^{\ast }(z,u;h)\subseteq \mathcal{K}%
\right\} )\leq F_{Y|Z=z}(\mathcal{K})\quad \text{ a.e. }z\in \mathcal{R}%
_{Z}\right\}
\end{multline}%
where $\mathsf{\ K}(\mathcal{R}_{Y})$ is the collection of closed sets on
the support of $Y$. In many cases arising in econometrics in
which distributional restrictions are put on $U$ this is less convenient to
work with than the characterization (\ref{Corollary 1 sharp set}). $Y^{\ast
} $ sets for a dynamic binary response two period panel model are shown in
Table \ref{Table: ystarsets 2BRP} for the case in which the initial value $%
Y_{0}$ is observed and for the case in which it is not.
\end{enumerate}

Some examples of the application of these results are now presented.%
\footnote{%
The development of some of these results was done by exploiting the symbolic
computational power of \textsf{Mathematica}, \cite{Mathematica}.}

\section{\label{Section: linear panel model}Linear panel data model}

The approach set out in this paper delivers classical results when taken to
the simple linear panel data model. Consider the simplest case with two
periods of observation and the following model incorporating a conditional
mean independence restriction

\begin{equation*}
Y_{t}=\beta _{0}+\beta _{1}Z_{t}+V+U_{t}\text{, } \quad \mathbb{E}[U_t|Z]=0 {%
, }\quad t\in \{1,2\},
\end{equation*}%
where $Z_{1}$ and $Z_{2}$ are scalar, $Z\equiv (Z_{1},Z_{2})$, and $z\equiv
(z_{1},z_{2})$.\footnote{%
The function 
\begin{equation*}
\sum_{t=1}^{T}\left( Y_{t}-\left( \beta _{0}+Z_{t}\beta _{1}+V+U_{t}\right)
\right) ^{2}
\end{equation*}%
can serve as the function $h(Y,Z,U,V)$.}

The $Y^{\ast }$ and $U^{\ast }$ sets are as follows.%
\begin{equation*}
\mathcal{Y}^{\ast }(u,z;\beta )=\{(y_{1},y_{2}): y_{2}-y_{1} =\beta
_{1}\left( z_{2}-z_{1}\right) +u_{2}-u_{1}\}
\end{equation*}%
\begin{equation*}
\mathcal{U}^{\ast }(y,z;\beta )=\{(u_{1},u_{2}):u_{2}-u_{1}= y_{2}-y_{1}
-\beta _{1}\left( z_{2}-z_{1}\right) \}
\end{equation*}

Theorem 5 of CR delivers the result that the values of $\beta _{1},$ say $%
\beta _{1}^{+}$ in the identified set are all values such that zero is an
element of the Aumann expectation of the set $\mathcal{U}^{\ast }(Y,Z;\beta
_{1}^{+})$ conditional on $Z=z$ for all $z\in \mathcal{R}_{Z}$. The set $%
\mathcal{U}^{\ast }(Y,Z;\beta _{1})$ is singleton in this example, so the
Aumann expectation is simply the classical expectation of point-valued
random variables and there is%
\begin{equation*}
\mathbb{E}[\mathcal{U}^{\ast }(Y,Z;\beta _{1})|Z=z]=\mathbb{E}%
[Y_{2}-Y_{1}|Z=z]-\beta _{1}\left( z_{2}-z_{1}\right)
\end{equation*}%
which, set equal to zero, delivers the correspondence%
\begin{equation*}
\beta _{1}=\frac{\mathbb{E}[Y_{2}-Y_{1}|Z=z]}{\left( z_{2}-z_{1}\right) }
\end{equation*}%
which is point identifying as long as $z_{2}\neq z_{1}$.

Extension to $T>2$ and dynamic models is straightforward and need not be
rehearsed here. The point is that the general approach proposed here
delivers classical results.

However the approach will \emph{not} deliver the well-known point
identification result in binary response panel data models with logistic independently distributed time-varying latent variables because those models further
impose $U 
\mathbin{\vbox{\baselineskip=0pt\lineskip=0pt
  \moveright2.5pt\hbox{$\|$}
  \hrule height 0.2pt width 10pt}} V$.\footnote{%
See \cite{chamberlain2010binary}.} In this paper the covariation of $V$ with
all other variables is unrestricted.

\section{Binary response panel models}

\label{Section: binary response}

This section studies the dynamic binary response model of Example \ref{ex 1}
under a variety of restrictions. Only in the final Section \ref{Section:
general binary panel model} are models admitting endogenous explanatory
variables considered. Section \ref{Section: literature} lists many papers
that study binary response panel models with fixed effects. In all but one
previous paper known to us there is a restriction on the joint distribution
of the fixed effect and other variables such that the conditional
distribution of other variables given the fixed effect is subject to
restrictions. No such restrictions are imposed here. The one exception of which we are aware is \cite{aristodemou2021semiparametric}, in which bounds are provided for binary response panel data models with an observed initial condition. 

Section \ref{Section: 2 period dynamic panel} gives results for the two
period dynamic binary response model when the initial condition ($Y_{0}$) is
observed. This model is studied in \cite{aristodemou2021semiparametric}. \
Three period dynamic models with unobserved initial condition are studied in
Section \ref{Section: 3 period dynamic y0 unobserved}. Section \ref{Section:
general binary panel model} gives results for a general case in which there
may be endogenous explanatory variables. Extension to models with multiple
lagged dependent variables is straightforward.

Define $Y=(Y_{1},\dots ,Y_{T})$ and $Z$ and $U$ similarly.

\subsection{\label{Section: 2 period dynamic panel}Two period dynamic binary
response model, initial condition observed}

In the case considered in this section, $T=2$ and $Y_{0}$ is observed.
Define $\Delta u\equiv u_{2}-u_{1}$, $\Delta z\equiv z_{2}-z_{1}$, and $%
\theta =(\beta ^{\prime },\gamma )^{\prime }$.

The $U^{\ast }$ sets are as follows.%\footnote{%
%The $Y^{\ast }$ sets are, for $\gamma \geq 0$:%
%\begin{equation*}
%\mathcal{Y}^{\ast }(u,z,y_{0};\theta )\left\{ 
%\begin{array}{ccc}
%\{(0,0),(0,1),(1,1)\} & , & \Delta u \geq -\Delta z\beta +y_{0}\gamma \\ 
%\{(0,0),(1,0),(1,1)\} & , & \Delta u \leq -\Delta z\beta +(y_{0}-1)\gamma%
%\end{array}%
%\right. ,
%\end{equation*}%
%and for $\gamma \leq 0$:%
%\begin{equation*}
%\mathcal{Y}^{\ast }(u,z,y_{0};\theta )\left\{ 
%\begin{array}{ccc}
%\{(0,0),(0,1),(1,1)\} & , & \Delta u \geq -\Delta z\beta +y_{0}\gamma \\ 
%\mathcal{R}_{Y} & , & -\Delta z\beta +y_{0}\gamma <\Delta u < -\Delta z\beta
%+(y_{0}-1)\gamma \\ 
%\{(0,0),(1,0),(1,1)\} & , & \Delta u \leq -\Delta z\beta +(y_{0}-1)\gamma%
%\end{array}%
%\right. .
%\end{equation*}%
%\par
%Only the $U^{\ast }$ sets are employed in producing identified sets and
%henceforth $Y^{\ast }$ sets are not reported.}%
\begin{equation*}
\mathcal{U}^{\ast }(y,z,y_{0};\theta )=\left\{ 
\begin{array}{ccc}
\mathcal{R}_{U} & , & y=(0,0) \\ 
\left\{ u:\Delta u \geq -\Delta z\beta +y_{0}\gamma \right\} & , & y=(0,1)
\\ 
\left\{ u:\Delta u \leq -\Delta z\beta +(y_{0}-1)\gamma \right\} & , & 
y=(1,0) \\ 
\mathcal{R}_{U} & , & y=(1,1)%
\end{array}%
\right.
\end{equation*}%
Unions of these $U^{\ast }$ sets do not deliver additional informative
inequalities.\footnote{%
Unions are either disjoint or equal to the support of $U$ depending on the
sign of $\gamma$.}

Under the independence restriction $U%
%TCIMACRO{\TeXButton{indt}{\indt}}%
%BeginExpansion
\indt%
%EndExpansion
Z|Y_{0}$ the identified set of values of $(\theta ,G_{U|Y_{0}})$ comprises
those values such that the following inequalities hold for $y_{0}\in \{0,1\}$
and a.e. $z\in \mathcal{R}_{Z}$. 
\begin{equation*}
\mathbb{P}[Y=(0,1)|Z=z,Y_{0}=y_{0}]\leq G_{U|Y_{0}=y_{0}}\left( \left\{
u:\Delta u \geq -\Delta z\beta +y_{0}\gamma \right\} \right)
\end{equation*}%
\begin{equation*}
\mathbb{P}[Y=(1,0)|Z=z,Y_{0}=y_{0}]\leq G_{U|Y_{0}=y_{0}}\left( \left\{
u:\Delta u \leq -\Delta z\beta +(y_{0}-1)\gamma \right\} \right)
\end{equation*}%
These are the inequalities of Theorem 1 of \cite%
{aristodemou2021semiparametric}. Setting $\gamma =0$ with $U 
\mathbin{\vbox{\baselineskip=0pt\lineskip=0pt
  \moveright2.5pt\hbox{$\|$}
  \hrule height 0.2pt width 10pt}} Z$, dropping conditioning on $Y_0$,
delivers the inequalities defining the identified set in the two period
static binary response panel model.

Table \ref{Table: ystarsets 2BRP} shows the $Y^{\ast }$ sets for the two
period dynamic binary response panel model. The top half of the table shows
the sets for the case in which $Y_{0}$ is observed. The bottom part shows
the sets obtained when $Y_{0}$ is not observed.%\footnote{As noted in Remark 9 in Section \ref{Section Identified sets}, identified sets can be characterized by conditions that guarantee the conditional distributions of $Y$ given $Z=z$ are selectionable with respect to the conditional distributions of random $Y^{\ast}$ sets $\mathcal{Y}^{\ast}(U,Z;h)$. See for example \cite{Beresteanu/Molchanov/Molinari:09} and \cite{Molinari:Handbook}.}

Appendix \ref{Appendix: Profile Bounds} derives sharp bounds on $\theta$ absent any specification of the distribution of $U$ using the method set out in Remark 8 of Section \ref{Section Identified sets}.

%TCIMACRO{\TeXButton{B}{\begin{table}[tbp] \centering}}%
%BeginExpansion
\begin{table}[tbp] \centering%
%EndExpansion
\caption{$Y^{\ast}$ sets in the binary response two period panel.}\medskip 
\begin{tabular}{|c|c|c|}
\hline
$Y_{0}$ & $\mathcal{Y}$ & $\left\{ u:\mathcal{Y}^{\ast }(z,u;\theta )=%
\mathcal{Y}\right\} $ \\ \hline\hline
& $\left\{ (0,0),(1,1)\right\} $ & $\left\{ u:-\Delta z\beta
+(y_{0}-1)\gamma <\Delta u<-\Delta z\beta +y_{0}\gamma \right\} $ \\ 
\cline{2-3}
observed & $\left\{ (0,0),(0,1),(1,1)\right\} $ & $\left\{ u:\Delta u\geq
-\Delta z\beta +y_{0}\gamma \wedge \Delta u>-\Delta z\beta +\left(
y_{0}-1\right) \gamma \right\} $ \\ \cline{2-3}\cline{2-3}
& $\left\{ (0,0),(1,0),(1,1)\right\} $ & $\left\{ u:\Delta u\leq -\Delta
z\beta +\left( y_{0}-1\right) \gamma \wedge \Delta u<-\Delta z\beta
+y_{0}\gamma \right\} $ \\ \cline{2-3}
& $\mathcal{R}_{Y}$ & $\left\{ -\Delta z\beta +y_{0}\gamma \leq \Delta u\leq
-\Delta z\beta +(y_{0}-1)\gamma \right\} $ \\ \hline\hline
not & $\left\{ (0,0),(0,1),(1,1)\right\} $ & $\left\{ u:\Delta u>\max \left(
-\Delta z\beta -\gamma ,-\Delta z\beta \right) \right\} $ \\ \cline{2-3}
observed & $\left\{ (0,0),\left( 1,0\right) ,(1,1)\right\} $ & $\left\{
u:\Delta u<\min \left( -\Delta z\beta +\gamma ,-\Delta z\beta \right)
\right\} $ \\ \cline{2-3}
& $\mathcal{R}_{Y}$ & $\left\{ u:\min \left( -\Delta z\beta +\gamma ,-\Delta
z\beta \right) \leq \Delta u\leq \max \left( -\Delta z\beta -\gamma ,-\Delta
z\beta \right) \right\} $ \\ \hline
\end{tabular}%
\label{Table: ystarsets 2BRP}%
%TCIMACRO{\TeXButton{E}{\end{table}}}%
%BeginExpansion
\end{table}%
%EndExpansion

\subsection{\label{Section: 3 period dynamic y0 unobserved}A three period
dynamic binary response model with the initial condition not observed}

For any $s,t \in [T]$ define $\Delta _{st}u\equiv u_{s}-u_{t}$, and $%
\Delta _{st}z\equiv z_{s}-z_{t}$. With $T=3$, and treating both $V$ and $%
Y_{0}$ as unobserved latent variables with unrestricted distributions the $%
U^{\ast }$ sets are as shown in Table \ref{Table: DBR3 ustar sets no Y0}.%Table \ref{Table: DBR3 ustar sets no Y0 pos gamma} ($\gamma \geq 0$) and Table \ref{Table: DBR3 ustar sets no Y0 neg gamma} ($\gamma \leq 0$).

%TCIMACRO{\TeXButton{B}{\begin{table}[tbp] \centering}}%
%BeginExpansion
\begin{table}[tbp] \centering%
%EndExpansion
\caption{$U^{\ast}$ sets in the dynamic binary response panel data model
with 3 periods and $Y_0$ not observed.}\medskip 
\begin{tabular}{|c|c|c|}
\hline
& $y$ & $\mathcal{U}^{\ast }(y,z;\theta )\quad $ \\ 
\hline\hline
1 & $(0,0,0)$ & $\mathcal{R}_{U}$ \\ \hline
2 & $(0,0,1)$ & $\{u:\left( \Delta _{31}u \geq -\Delta _{31}z\beta + \min \left(\gamma,0 \right) \right)
\wedge \left( \Delta _{32}u\geq-\Delta _{32}z\beta \right) \}$ \\ \hline
3 & $(0,1,0)$ & $\{u:\left( \Delta _{21}u\geq-\Delta _{21}z\beta + \min \left(\gamma,0 \right) \right)
\wedge \left( \Delta _{32}u\leq-\Delta _{32}z\beta -\gamma \right) \}$ \\ 
\hline
4 & $(0,1,1)$ & $\{u:\left( \Delta _{21}u\geq-\Delta _{21}z\beta + \min \left(\gamma,0 \right)  \right)
\wedge (\Delta _{31}u\geq-\Delta _{31}z\beta - \max \left(\gamma,0 \right) )\}$ \\ \hline
5 & $(1,0,0)$ & $\{u:\left( \Delta _{21}u\leq-\Delta _{21}z\beta - \min\left(\gamma,0 \right) \right)
\wedge \left( \Delta _{31}u\leq-\Delta _{31}z\beta +\max\left(\gamma,0 \right) \right) \}$ \\ 
\hline
6 & $(1,0,1)$ & $\{u:\left( \Delta _{21}u\leq-\Delta _{21}z\beta - \min\left(\gamma,0 \right)\right)
\wedge \left( \Delta _{32}u\geq-\Delta _{32}z\beta +\gamma \right) \}$ \\ 
\hline
7 & $(1,1,0)$ & $\{u:\left( \Delta _{31}u\leq-\Delta _{31}z\beta -\min\left(\gamma,0 \right)\right)
\wedge \left( \Delta _{32}u\leq-\Delta _{32}z\beta \right) \}$ \\ \hline
8 & $(1,1,1)$ & $\mathcal{R}_{U}$ \\ \hline
\end{tabular}%
\label{Table: DBR3 ustar sets no Y0}%
%TCIMACRO{\TeXButton{E}{\end{table}}}%
%BeginExpansion
\end{table}%

For sets of values of $Y$, $\mathcal{T\subset R}_{Y}$, define functions 
\begin{equation}
\mathcal{S}(\mathcal{T},z;\theta )\equiv \dbigcup\limits_{y\in \mathcal{T}}%
\mathcal{U}^{\ast }(y,z;\theta )  \label{Sfunction}
\end{equation}
and\footnote{The set $\mathcal{T}$ can be a strict subset of $\mathcal{Y}(\mathcal{T},z;\theta )$. For example, this is the case when $\mathcal{T}$ contains two values of $Y$ and there is a third value of $Y$ such that its $U^{\ast }$ set is a subset of $\mathcal{S}(\mathcal{T},z;\theta )$ as in row 7 of Table \ref{Table:SBP3 inequalities}.}
\begin{equation}
\mathcal{Y}(\mathcal{T},z;\theta )\equiv \{y:\mathcal{U}^{\ast }(y,z;\theta
)\subseteq \mathcal{S}(\mathcal{T},z;\theta )\}.  \label{Yfunction}
\end{equation}%
The identified set of values of $\left( \beta ,\gamma ,\mathcal{G}_{U|Z=z}\right) $ comprises the values satisfying inequalities of the form
\begin{equation*}
\mathbb{P}[Y\in \mathcal{Y}(\mathcal{T},z;\theta )|Z=z]\leq G_{U|Z=z}(\mathcal{S}
(\mathcal{T},z;\theta ))\text{, a.e. } z\in \mathcal{R}_{Z}\text{,}
\end{equation*}
where the sets $\mathcal{Y}(\mathcal{T},z;\theta )$ and $\mathcal{T}$ are
shown in the first and second columns of Tables \ref{Table:SBP3 inequalities}, \ref{Table: DBP3 core determining inequalities gamma positive}, and \ref{Table: DBP3 inequalities gamma negative}, covering the cases in which $\gamma = 0$, $\gamma >0$, and $\gamma <0$, respectively. %If $\mathcal{G}_{U|Z=z}$ is restricted such that $U \indt Z$ so that $G_{U|Z=z} = G_U$ for all $z$, the identified set for $\left( \beta ,\gamma ,G_U\right)$ are those satisfying
%\begin{equation*}
%\mathbb{P}[Y\in \mathcal{Y}(\mathcal{T},z;\theta )|Z=z]\leq G_{U}(\mathcal{S}
%(\mathcal{T},z;\theta ))\text{, a.e. } z\in \mathcal{R}_{Z}\text{,}
%\end{equation*}
%for the same combinations of $\mathcal{Y}(\mathcal{T},z;\theta )$ and $\mathcal{T}$ specified in the tables.

\subsection{\label{Section: general binary panel model}General dynamic
binary response panel models}

Consider now the general specification of a dynamic panel data model from
Example 1, allowing for endogeneity admitting $\alpha \neq 0$. This section illustrates application of our
identification analysis to such cases, also allowing for arbitrary finite $T$%
.\footnote{Here we impose $\mathcal{R}_C=\mathbb{R}$, as typically done in the
literature. Extension to cases in which $\mathcal{R}_C$ is a subset of $%
\mathbb{R}$ is straightforward.}

Define 
\begin{equation}
\mathcal{T}_{0}\equiv \left\{ t\in [T] :Y_{1t}=0\right\} 
\text{,}\qquad \mathcal{T}_{1}\equiv \left\{ t\in [T]
:Y_{1t}=1\right\} \text{,}  \label{set T choice definitions}
\end{equation}
denoting the sets of periods in which $Y_{1t}=0$ and $Y_{1t}=1$,
respectively. Let $\mathcal{Y}_{0}$ denote the set of values in which the
initial condition $Y_{10}$ is known to lie, with $\mathcal{Y}%
_{0}=\left\{ Y_{10}\right\} $ if the initial condition is observed and $%
\mathcal{Y}_{0}=\{0,1\}$ if the initial condition is not observed.

The set $\mathcal{U}^{\ast }(Y,Z;h)$ defined in (\ref{Ustar set}) in
this model can be written 
\begin{multline}  \label{dynamic binary U set}
\mathcal{U}^{\ast }(Y,Z;h)=\Bigl\{u\in \mathcal{R}_{U}:\exists Y_{10}\in \mathcal{%
Y}_{0}\text{ such that } \\
\max_{t\in \mathcal{T}_{0}} \{ Y_{2t}\alpha +Z_{t}\beta +Y_{1t-1}\gamma
+u_{t}\}\leq \min_{t\in \mathcal{T}_{1}}\{Y_{2t}\alpha +Z_{t}\beta
+Y_{1t-1}\gamma +u_{t}\}\Bigr\}\text{.}
\end{multline}
This is so because the constituent inequalities may be equivalently
expressed as 
\begin{equation*}
\underline{C}\leq \overline{C}
\end{equation*}%
where 
\begin{align*}
\underline{C}& \equiv \max_{t\in \mathcal{T}_{1}}\bigl\{-\left( Y_{2t}\alpha
+Z_{t}\beta +Y_{1t-1}\gamma +u_{t}\right)\bigr\} \text{,} \\
\overline{C}& \equiv \min_{t\in \mathcal{T}_{0}}\bigl\{-\left( Y_{2t}\alpha
+Z_{t}\beta +Y_{1t-1}\gamma +u_{t}\right)\bigr\} \text{.}
\end{align*}%
That $\underline{C}\leq \overline{C}$ for some $Y_{10}\in \mathcal{Y}_{0}$
guarantees there exist values $C\in \left[ \underline{C},\overline{C}\right] $
and $Y_{10}\in \mathcal{Y}_{0}$ such that (\ref{example eq}) holds.\footnote{%
The $\max$ and $\min$ operators applied to the empty set are defined to be $%
-\infty$ and $\infty$, respectively.}

Define $\theta =(\alpha ^{\prime },\beta ^{\prime },\gamma )^{\prime }$. For
any panel data model for a binary outcome as in (\ref{example eq}) with $%
U\sim G_{U}$ independent of $Z$, the identified set of values of $\left(
\theta ,G_{U}\right) $ are those pairs satisfying, for an appropriately
chosen collection\footnote{%
The collection of all unions of $U^{\ast }$ sets, $\mathsf{U}^{\ast }(z;h)$,
defined in (\ref{all unions definition}), will suffice. In practice there
may be unions in this collection which need not be considered because they
deliver redundant inequalities.} of sets $\mathcal{T}$, the inequalities 
\begin{equation*}
\mathbb{P}[Y\in \mathcal{Y}(\mathcal{T},z;\theta )|Z=z]\leq G_{U}(\mathcal{S}%
(\mathcal{T},z;\theta ))\text{, a.e. }z \in \mathcal{R}_Z \text{,}
\end{equation*}%
where the sets $\mathcal{S}(\mathcal{T},z;\theta )$ and $\mathcal{Y}(%
\mathcal{T},z;\theta )$ are as defined in (\ref{Sfunction}) and (\ref%
{Yfunction}).

This characterization applies for dynamic models and static models (for
which $\gamma =0$ is imposed), models allowing endogenous explanatory
variables (for which $\alpha \neq 0$ is permitted), and for arbitrary $T$.

\section{\label{Section: Multiple discrete choice}Static multiple discrete
choice panel models}

In this section multiple discrete choice panel models are considered.  The presence of the fixed effect renders the model incomplete as in the multiple discrete choice analysis of \cite{Chesher/Rosen/Smolinski:11MNLIV}. In that analysis incompleteness arose due to the inclusion of potentially endogenous explanatory variables. Analysis of a static panel model with $T=2$ periods is considered.

In a three-choice model with two periods there is
\begin{equation*}
Y_{t}=\argmax_{d}\{J_{dt}:d\in \{1,2,3\}\},\quad t\in \{1,2\}
\end{equation*}%
where the $J_{dt}$ terms are random utilities with parameters $\theta \equiv (\beta _{1}^{\prime },\beta _{2}^{\prime
})^{\prime }$ as follows.%
\begin{equation*}
J_{1t}\equiv Z_{t}\beta _{1}+V_{1}+U_{1t},\quad t\in \{1,2\}\text{,}
\end{equation*}%
\begin{equation*}
J_{2t}\equiv Z_{t}\beta _{2}+V_{2}+U_{2t},\quad t\in \{1,2\}\text{,}
\end{equation*}%
\begin{equation*}
J_{3t}\equiv U_{3t},\quad t\in \{1,2\}\text{.}
\end{equation*}%
The terms $V_{1}$ and $V_{2}$ are \textquotedblleft fixed
effects\textquotedblright\ whose distribution and covariation with other
variables is unrestricted.

Section \ref{Section: literature} lists many papers that study multiple
discrete panel models with fixed effects. In all studies of multiple
discrete choice panel data models known to us there are conditions imposed
on the joint distribution of fixed effects and other variables such that the
conditional distribution of other variables given the fixed effect is
subject to restriction. No such restrictions are imposed here.

%Define $W_{1t}\equiv U_{1t}-U_{3t}$, $W_{2t}\equiv U_{2t}-U_{3t}$, $W\equiv (W_{11},W_{12},W_{21},W_{22})$, $Z\equiv (Z_{1},Z_{2})$, $\Delta z\equiv z_{2}-z_{1}$, $\Delta w_{1}\equiv w_{12}-w_{11}$, $\Delta w_{2}\equiv w_{22}-w_{21}$, and $\theta \equiv (\beta _{1}^{\prime },\beta _{2}^{\prime })^{\prime }$.

The $U^{\ast }$ sets are shown in Table \ref{Table: MDC ustar sets} using notation $\Delta U_d \equiv U_{d2} - U_{d1}$.

%TCIMACRO{\TeXButton{B}{\begin{table}[tbp] \centering}}%
%BeginExpansion
\begin{table}[tbp] \centering%
%EndExpansion
\caption{$U^{\ast}$ sets  in the multiple discrete three choice two period panel
model.}\medskip 
\begin{tabular}{|c|c|c|}
\hline
& $y$ & $\mathcal{U}^{\ast }(y,z;\theta )$ \\ \hline\hline
1 & $(1,1)$ & $\mathcal{R}_{U}$ \\ \hline
2 & $(1,2)$ & $\{u:\Delta u_{2}-\Delta u_{1} \geq \Delta z\beta _{1}-\Delta
z\beta _{2}\}$ \\ \hline
3 & $(1,3)$ & $\{u:\Delta u_{1} - \Delta u_3 \leq -\Delta z\beta _{1}\}$ \\ \hline
4 & $(2,1)$ & $\{u:\Delta u_{2}-\Delta u_{1} \leq \Delta z\beta _{1}-\Delta
z\beta _{2}\}$ \\ \hline
5 & $(2,2)$ & $\mathcal{R}_{U}$ \\ \hline
6 & $(2,3)$ & $\{u:\Delta u_{2} - \Delta u_3 \leq -\Delta z\beta _{2}\}$ \\ \hline
7 & $(3,1)$ & $\{u:\Delta u_{1}-\Delta u_{3} \geq -\Delta z\beta _{1}\}$ \\ \hline
8 & $(3,2)$ & $\{u:\Delta u_{2} - \Delta u_3 \geq -\Delta z\beta _{2}\}$ \\ \hline
9 & $(3,3)$ & $\mathcal{R}_{U}$ \\ \hline
\end{tabular}%
\label{Table: MDC ustar sets}%
%TCIMACRO{\TeXButton{E}{\end{table}}}%
%BeginExpansion
\end{table}%
%EndExpansion

The identified set of values of $\left( \theta ,G_{U}\right) $ are those
pairs satisfying, for all $z\in \mathcal{R}_{Z}$ the inequalities%
\begin{equation*}
\mathbb{P}[Y\in \mathcal{Y}(\mathcal{T},z;\theta )|Z=z]\leq G_{U|Z=z}(\mathcal{S}%
(\mathcal{T},z;\theta )), \quad \text{a.e. } z \in \mathcal{R}_Z\text{,}
\end{equation*}%
where the sets $\mathcal{S}(\mathcal{T},z;\theta )$ and $\mathcal{Y}(%
\mathcal{T},z;\theta )$ are as defined in (\ref{Sfunction}) and (\ref%
{Yfunction}) and the sets $\mathcal{T}$ and $\mathcal{Y}(\mathcal{T}
,z;\theta )$ are shown in Table \ref{Table: MDC inequalities using T sets}.
As we show for ordered choice panels in the following section, this
characterization can be generalized to allow arbitrary periods $T$ and
alternatives $\left\{1,...,K\right\}$, and can allow for dependence on
lagged choices. Endogenous covariates can be permitted as is done for cross
sectional multiple discrete choice in \cite{Chesher/Rosen/Smolinski:11MNLIV}.

\section{Ordered response panel models}

\label{Section: Ordered choice panel models} This section generalizes the
binary response models of Section \ref{Section: binary response} to models
in which the outcome is an ordered response variable. Section \ref%
{subsection: two period three choice ordered response} gives results for a
static two period model with three ordered outcomes. Section \ref%
{subsection: General ordered response panel model} then gives results for a
general ordered outcome model allowing an arbitrary finite number of ordered
outcomes, arbitrary periods, and dynamics.

\subsection{Two period ordered response panel models with three categories}

\label{subsection: two period three choice ordered response}

There are structural equations as follows. 
\begin{equation*}
Y_{t}=\left\{ 
\begin{array}{ccc}
0 & , & Z_{t}\beta +V+U_{t}\leq c_{1} \\ 
1 & , & c_{1}\leq Z_{t}\beta +V+U_{t}\leq c_{2} \\ 
2 & , & c_{2}\leq Z_{t}\beta +V+U_{t}%
\end{array}%
\right. ,\quad t\in \{1,2\}
\end{equation*}%
Let $Y=(Y_{1},Y_{2})$, $Z=(Z_{1},Z_{2})$, $U=(U_{1},U_{2})$. Let $\theta
=(\beta ^{\prime },c_{1},c_{2})^{\prime }$.\footnote{%
In some applications $c_{1}$ and $c_{2}$ can have known values.} There is
the restriction $U%
%TCIMACRO{\TeXButton{indt}{\indt}}%
%BeginExpansion
\indt%
%EndExpansion
Z$. This model is studied in \cite{aristodemou2021semiparametric} where, as
here, no restrictions are placed on the distribution of $V$ or on its
covariation with other variables.

%TCIMACRO{\TeXButton{B}{\begin{table}[tbp] \centering}}%
%BeginExpansion
\begin{table}[tbp] \centering%
%EndExpansion
\caption{$U^{\ast}$ sets  in the ordered response three category  two period panel
model.}\medskip 
\begin{tabular}{|c|c|c|}
\hline
& $y$ & $\mathcal{U}^{\ast }(y,z;\theta )$ \\ \hline\hline
1 & $(0,0)$ & $\mathcal{R}_{U}$ \\ \hline
2 & $(0,1)$ & $\mathcal{\{}u:\Delta u\geq -\Delta z\beta \}$ \\ \hline
3 & $(0,2)$ & $\{u:\Delta u\geq c_{2}-c_{1}-\Delta z\beta \}$ \\ \hline
4 & $(1,0)$ & $\mathcal{\{}u:\Delta u\leq -\Delta z\beta \}$ \\ \hline
5 & $(1,1)$ & $\mathcal{\{}u:\left( \Delta u\leq c_{2}-c_{1}-\Delta z\beta
\right) \wedge \left( \Delta u\geq c_{1}-c_{2}-\Delta z\beta \right) \}$ \\ 
\hline
6 & $(1,2)$ & $\mathcal{\{}u:\Delta u\geq -\Delta z\beta \}$ \\ \hline
7 & $(2,0)$ & $\mathcal{\{}u:\Delta u\leq c_{1}-c_{2}-\Delta z\beta \}$ \\ 
\hline
8 & $(2,1)$ & $\mathcal{\{}u:\Delta u\leq -\Delta z\beta \}$ \\ \hline
9 & $(2,2)$ & $\mathcal{R}_{U}$ \\ \hline
\end{tabular}%
\label{Table: OC3 ustar sets}%
%TCIMACRO{\TeXButton{E}{\end{table}}}%
%BeginExpansion
\end{table}%
%EndExpansion

Define $\Delta u\equiv u_{2}-u_{1}$ and $\Delta z\equiv z_{2}-z_{1}$. The $%
U^{\ast }$ sets are shown in Table \ref{Table: OC3 ustar sets}. The
identified set of values of $(\theta ,G_{U})$ comprises the values
satisfying, for $z\in \mathcal{R}_{Z}$, $7$ inequalities of the form%
\begin{equation*}
\mathbb{P}[Y\in \mathcal{Y}|Z=z]\leq G_{U}(\mathcal{S})
\end{equation*}%
where $\mathcal{Y}$ and $\mathcal{S}$ are given in Table \ref{Table: OC3 inequalities}.

%TCIMACRO{\TeXButton{B}{\begin{table}[tbp] \centering}}%
%BeginExpansion
\begin{table}[tbp] \centering%
%EndExpansion
\caption{Sets $\mathcal{Y}$ and $\mathcal{S}$ in the inequalities defining
the identified set of values of $\beta$ and $G_U$ in the two period ordered response panel model with three categories.}%
\medskip 
\begin{tabular}{|c|c|c|}
\hline
& $\mathcal{Y}$ & $\mathcal{S}$ \\ \hline\hline
1 & $\left\{ (0,2)\right\} $ & $\{u:\Delta u \geq c_{2}-c_{1}-\Delta z\beta
\}$ \\ \hline
2 & $\left\{ (1,1)\right\} $ & $\{u:\left( \Delta u \leq c_{2}-c_{1}-\Delta
z\beta \right) \wedge \left( \Delta u \geq c_{1}-c_{2}-\Delta z\beta \right)
\}$ \\ \hline
3 & $\left\{ (2,0)\right\} $ & $\{u:\Delta u\leq c_{1}-c_{2}-\Delta z\beta
\} $ \\ \hline
4 & $\left\{ (0,1),(0,2),(1,2)\right\} $ & $\{u:\Delta u \geq -\Delta z\beta
\}$ \\ \hline
5 & $\left\{ (1,0),(2,0),(2,1)\right\} $ & $\{u:\Delta u \leq -\Delta z\beta
\}$ \\ \hline
6 & $\{(0,1),(0,2),(1,1),(1,2)\}$ & $\{u:\Delta u \geq c_{1}-c_{2}-\Delta
z\beta \}$ \\ \hline
7 & $\{(1,0),(1,1),(2,0),(2,1)\}$ & $\{u:\Delta u \leq c_{2}-c_{1}-\Delta
z\beta \}$ \\ \hline
\end{tabular}%
\label{Table: OC3 inequalities}%
%TCIMACRO{\TeXButton{E}{\end{table}}}%
%BeginExpansion
\end{table}%
%EndExpansion

Theorem 5 of \cite{aristodemou2021semiparametric} delivers an outer set
using the inequalities 1, 2 and 3 in Table \ref{Table: OC3 inequalities} and
the inequalities:%
\begin{equation*}
\mathbb{P}(Y=(0,1)|Z=z]\leq G_{U}(\mathcal{\{}u:\Delta u>-\Delta z\beta \})
\end{equation*}%
and%
\begin{equation*}
\mathbb{P}(Y=(1,2)|Z=z]\leq G_{U}(\mathcal{\{}u:\Delta u>-\Delta z\beta \})
\end{equation*}%
which are implied by inequality 4, and 
\begin{equation*}
\mathbb{P}(Y=(1,0)|Z=z]\leq G_{U}(\mathcal{\{}u:\Delta u<-\Delta z\beta \})
\end{equation*}%
and%
\begin{equation*}
\mathbb{P}(Y=(2,1)|Z=z]\leq G_{U}(\mathcal{\{}u:\Delta u<-\Delta z\beta \})
\end{equation*}%
which are implied by inequality 5.

\subsection{General ordered response panel models}

\label{subsection: General ordered response panel model} Consider now a
general specification of an ordered response panel data model with $\mathcal{%
R}_Y = \left\{0,...,J\right\}$ and allowing for dynamics as in e.g. \cite%
{honore2021dynamic} in which for all $j \in \mathcal{R}_Y$: 
\begin{equation}  \label{general ordered response equation}
Y_t = j \implies c_j \leq Z_{t}\beta + \imath_t \gamma + V + U_{t} \leq
c_{j+1}\text{,}
\end{equation}
where $c_0 \equiv -\infty$, $c_{J+1} \equiv \infty$, and $\imath_t \equiv
\left(1\left[Y_{t-1}=0 \right],\ldots,1\left[Y_{t-1}=J \right]\right)$ with
each component of $\gamma$ encoding the impact of lagged $Y$ on $Y_t$.%
\footnote{%
It is straightforward to accommodate multiple lags.} Let $\tilde{Z}_t \equiv
(Z_t,\imath_t)$, $\tilde{\beta} \equiv (\beta^{\prime},
\gamma^{\prime})^{\prime} $, $Y \equiv (Y_1,...,Y_T)$, $Z \equiv
(Z_1,...,Z_T)$, $U\equiv(U_1,...,U_T)$. Let $\theta \equiv
(\beta^{\prime},\gamma^{\prime},c_1,...,c_J)^{\prime }$ denote parameters of
the structural function, restricted such that $c_1 < \dots < c_J$. The
initial condition $Y_0$ is assumed observed, but it is straightforward to
accommodate an unobserved initial condition as for the binary panel studied
in Section \ref{Section: 3 period dynamic y0 unobserved}. 
%As before there is the restriction $U\indtZ$.\footnote{One could alternatively restrict $U \indt Z|Y_0$ as done in Section \ref{Section: 2 period dynamic panel} and proceed accordingly.} 

Sets $\mathcal{U}^{\ast}\left(Y,Z;h\right)$ are given by 
\begin{multline*}
\mathcal{U}^{\ast }(Y,Z;h)=\{ u\in \mathcal{R}_{U}:\forall s,t \in
[T], \\
u_t - u_s \leq c_{Y_t+1}-c_{Y_s}- (\tilde{Z}_t - \tilde{Z}_s)\tilde{\beta}
\} \text{,}
\end{multline*}
This is verified by noting that for all $u \in \mathcal{U}^{\ast }(Y,Z;h)$
we have that 
\begin{equation*}
\forall s,t \in [T], \qquad c_{Y_s} - \tilde{Z}_s \tilde{\beta} -
u_s \leq c_{Y_t+1} - \tilde{Z}_t \tilde{\beta} - u_t \text{,}
\end{equation*}
in turn implying the existence of $v$ such that 
\begin{equation*}
\forall s,t \in [T], \qquad c_{Y_s} - \tilde{Z}_s \tilde{\beta} -
u_s \leq v \leq c_{Y_t+1} - \tilde{Z}_t \tilde{\beta} - u_t \text{.}
\end{equation*}
For all such $u,v$ it follows that (\ref{general ordered response equation})
holds for all $t$ with $U=u$ and $V=v$.

When the independence restriction $U 
\mathbin{\vbox{\baselineskip=0pt\lineskip=0pt
  \moveright2.5pt\hbox{$\|$}
  \hrule height 0.2pt width 10pt}} Z$ is imposed, the identified set for $%
\left( \theta ,G_{U}\right) $ are those pairs satisfying 
\begin{equation*}
\mathbb{P}[Y\in \mathcal{Y}(\mathcal{T},z;\theta )|Z=z]\leq G_{U}(\mathcal{S}%
(\mathcal{T},z;\theta ))\text{, a.e. } z\in \mathcal{R}_Z
\end{equation*}
for an appropriately chosen collection of sets $\mathcal{T}$ where the sets $%
\mathcal{S}(\mathcal{T},z;\theta )$ and $\mathcal{Y}(\mathcal{T},z;\theta )$
are as defined in (\ref{Sfunction}) and (\ref{Yfunction}).\footnote{%
Once again the collection of all unions of $U^{\ast }$ sets, $\mathsf{U}%
^{\ast }(z;h)$, defined in (\ref{all unions definition}), will suffice, but
in practice some of these unions may not be necessary.} This
characterization can be generalized to allow for endogenous variables on the
right hand side of (\ref{general ordered response equation}) as done for
cross section analysis of ordered choice models in \cite%
{Chesher/Smolinski:12} and \cite{Chesher/Rosen/Siddique:23}. It is
straightforward to allow $G_{U|Z=z}$ to vary with $z$ by replacing $G_U$
with $G_{U|Z=z}$ in the inequality above, which then delivers an identified
set for pairs $\left( \theta ,\mathcal{G}_{U|Z}\right) $.

\section{\label{Section: simultaneous binary response}Simultaneous binary
response panel models}

There is the model%
\begin{eqnarray*}
Y_{1t} &=&1[\alpha _{1}Y_{2t}+Z_{t}\beta _{1}+V_{1}+U_{1t}\geq 0] \\
Y_{2t} &=&1[\alpha _{2}Y_{1t}+Z_{t}\beta _{2}+V_{2}+U_{2t}\geq 0]
\end{eqnarray*}%
with $t\in [T]$ and the independence restriction $(U_1,U_2)
%TCIMACRO{\TeXButton{indt}{\indt}}%
%BeginExpansion
\indt%
%EndExpansion
Z\equiv (Z_{1},\dots ,Z_{T})$ where for each $j \in  \{1,2\}$, $U_j \equiv (U_{j1},\dots ,U_{jT})$.\footnote{This strong exogeneity restriction can be relaxed.}

This is a simultaneous equations model with binary outcomes such as is found
in simultaneous firm entry applications\footnote{%
See for example \cite{tamer2003incomplete}.} and models of social
interactions, put into a panel context with \textquotedblleft fixed
effects\textquotedblright , constant through time, one for each outcome.

\cite{honore2021identification} study a restricted version of this model
with $\beta _{1}=\beta _{2}$, $\alpha _{1}=\alpha _{2}$ and $U$ and $V$
restricted to be independently distributed. No such restrictions are imposed
here.

Define $\theta \equiv (\alpha _{1},\alpha _{2},\beta _{1}^{\prime },\beta
_{2}^{\prime })^{\prime }$. The distribution of $V\equiv (V_{1},V_{2})$ and
the covariation of $V$ with other variables is unrestricted.

Consider the case with $T=2$ when $Y=(Y_{11},Y_{12},Y_{21},Y_{22})$.
Extension to more time periods and outcomes is straightforward.

Define $\Delta u_{1}\equiv u_{12}-u_{11}$, $\Delta u_{2}\equiv u_{22}-u_{21}$%
, $\Delta z\equiv z_{2}-z_{1}$. The $U^{\ast }$ sets, $\mathcal{U}^{\ast
}(y,z;\theta )$, are as shown in Table \ref{Table: SES ustar sets}.

%TCIMACRO{\TeXButton{B}{\begin{table}[tbp] \centering}}%
%BeginExpansion
\begin{table}[tbp] \centering%
%EndExpansion
\caption{$U^{\ast}$ sets in the simultaneous binary response two period
panel.}\medskip 
\begin{tabular}{|c|c|c|}
\hline
& $y$ & $\mathcal{U}^{\ast }(y,z;\theta )$ \\ \hline\hline
1 & $(0,0,0,0)$ & $\mathcal{R}_{U}$ \\ \hline
2 & $(0,0,0,1)$ & $\{u:\Delta u_{2} \geq -\Delta z\beta _{2}\}$ \\ \hline
3 & $(0,0,1,0)$ & $\{u:\Delta u_{2} \leq -\Delta z\beta _{2}\}$ \\ \hline
4 & $(0,0,1,1)$ & $\mathcal{R}_{U}$ \\ \hline
5 & $(0,1,0,0)$ & $\{u:\Delta u_{1}\geq -\Delta z\beta _{1}\}$ \\ \hline
6 & $(0,1,0,1)$ & $\{u:\left( \Delta u_{1} \geq -\Delta z\beta _{1}-\alpha
_{1}\right) \wedge \left( \Delta u_{2} \geq -\Delta z\beta _{2}-\alpha
_{2}\right) \}$ \\ \hline
7 & $(0,1,1,0)$ & $\{u:\left( \Delta u_{1} \geq -\Delta z\beta _{1}+\alpha
_{1}\right) \wedge \left( \Delta u_{2} \leq -\Delta z\beta _{2}-\alpha
_{2}\right) \}$ \\ \hline
8 & $(0,1,1,1)$ & $\{u:\Delta u_{1} \geq -\Delta z\beta _{1}\}$ \\ \hline
9 & $(1,0,0,0)$ & $\{u:\Delta u_{1} \leq -\Delta z\beta _{1}\}$ \\ \hline
10 & $(1,0,0,1)$ & $\{u:\left( \Delta u_{1} \leq -\Delta z\beta _{1}-\alpha
_{1}\right) \wedge \left( \Delta u_{2} \geq -\Delta z\beta _{2}+\alpha
_{2}\right) \}$ \\ \hline
11 & $(1,0,1,0)$ & $\{u:\left( \Delta u_{1} \leq -\Delta z\beta _{1}+\alpha
_{1}\right) \wedge \left( \Delta u_{2} \leq -\Delta z\beta _{2}+\alpha
_{2}\right) \}$ \\ \hline
12 & $(1,0,1,1)$ & $\{u:\Delta u_{1} \leq -\Delta z\beta _{1}\}$ \\ \hline
13 & $(1,1,0,0)$ & $\mathcal{R}_{U}$ \\ \hline
14 & $(1,1,0,1)$ & $\{u:\Delta u_{2} \geq -\Delta z\beta _{2}\}$ \\ \hline
15 & $(1,1,1,0)$ & $\{u:\Delta u_{2} \leq -\Delta z\beta _{2}\}$ \\ \hline
16 & $(1,1,1,1)$ & $\mathcal{R}_{U}$ \\ \hline
\end{tabular}%
\label{Table: SES ustar sets}%
%TCIMACRO{\TeXButton{E}{\end{table}}}%
%BeginExpansion
\end{table}%
%EndExpansion

There are $12$ $U^{\ast }$ sets that are not equal to $\mathcal{R}_{U}$ and $%
4$ pairs of these $U^{\ast }$ sets are identical - for example $\mathcal{U}%
^{\ast }(0,0,0,1),z;\theta )=\mathcal{U}^{\ast }(1,1,0,1),z;\theta )$, so
there are unions of $8$ $U^{\ast }$ sets to be considered when calculating
the identified set, that is $254$ unions in total. Only $24$ of
these deliver inequalities that characterize the identified set of parameter
values, the remaining unions delivering redundant inequalities.

The configuration of the unions of these $U^{\ast }$ sets depends on the
signs of $\alpha _{1}$ and $\alpha _{2}$ and in practice there are likely to
be restrictions on these. For example in a simultaneous firm entry
application $\alpha _{1}\leq 0$ and $\alpha _{2}\leq 0$ would likely be
imposed and in a model of couple's choices of activity (e.g. cinema
attendance) $\alpha _{1}\geq 0$ and $\alpha _{2}\geq 0$.

Only the case with $\alpha _{1}\geq 0$ and $\alpha _{2}\geq 0$ is presented
here. In this case, among the $U^{\ast }$ sets only the sets $\mathcal{U}%
^{\ast }((0,1,0,1),z;\theta )$ and $\mathcal{U}^{\ast }((1,0,1,0),z;\theta )$
have a non-empty intersection.

The identified set of values of $\left( \theta ,G_{U}\right) $ are those
pairs satisfying, for all $z\in \mathcal{R}_{Z}$ the inequalities%
\begin{equation*}
\mathbb{P}[Y\in \mathcal{Y}(\mathcal{T},z;\theta )|Z=z]\leq G_{U}(\mathcal{S}%
(\mathcal{T},z;\theta ))
\end{equation*}%
where the sets $\mathcal{S}(\mathcal{T},z;\theta )$ and $\mathcal{Y}(%
\mathcal{T},z;\theta )$ are as defined in (\ref{Sfunction}) and (\ref%
{Yfunction}) and the sets $\mathcal{T}$ and $\mathcal{Y}(\mathcal{T}%
,z;\theta )$ are shown in Table \ref{Table:SES inequalities}.

\section{Censored outcome panels}\label{Section: Tobit panels}

In a panel model with a censored outcome there is the following.
\begin{equation*}
Y_{1t}=\max (\alpha Y_{2t}+Z_{t}\beta +\gamma Y_{1t-1}+C+U_{t},Y_{3t})\text{,}\quad t\in [T]\text{,}\quad C\in \mathcal{%
%TCIMACRO{\U{211d} }%
%BeginExpansion
\mathbb{R}
%EndExpansion
}\text{,}\quad Y_{10} \in \mathcal{Y}_{10}\text{,}
\end{equation*}
with $V \equiv (C, Y_{10})$ and $Y_{3t}$ denoting a censoring threshold, such that the outcome variable $Y_{1t}$ takes the value of the index $\alpha Y_{2t}+Z_{t}\beta +\gamma Y_{1t-1}+C+U_{t}$ when it exceeds the censoring threshold, and otherwise takes the value $Y_{3t}$. The censoring indicator $W_t \equiv 1\left[ Y_{1t}=Y_{3t} \right]$ is observed.

As in the models studied in KPT, the censoring threshold $Y_{3t}$ can be endogenous, and it may be correlated with elements of $U$ and $V$.\footnote{It is not necessary for $Y_{3t}$ to be observed in periods without censoring.} As before, endogenous $Y_{2t}$ is permitted in models with $\alpha \neq 0$, as in cross-sectional Tobit models studied in \cite{Chesher/Kim/Rosen:23}. The set $\mathcal{Y}_{10}$ denotes the feasible set of values for the initial condition $Y_{10}$ given the observed variables.\footnote{So $\mathcal{Y}_{10}$ is the singleton $\{Y_{10}\}$ if the initial condition is observed, and would typically be its entire support if it is not observed. In a static model the analysis applies with $\gamma =0$ and $Y_{10}$ absent.}

Define 
\begin{equation*}
\mathcal{T}_{0} \equiv \{t\in [T]:W_{t}=0\}\text{,} \qquad \mathcal{T}_{1} \equiv \{t\in [T]:W_{t}=1\}\text{.}
\end{equation*}
Adopting the strategy for obtaining $U^{\ast}$ sets described in remark 5 of Section \ref{Section Identified sets} we can define
\begin{equation*}
\mathcal{D}_{t}(y,z,u;h)=\left\{ 
\begin{array}{ccc}
\{(c,y_{10}) \in \mathbb{R}\times \mathcal{Y}_{10} : c = y_{1t}-\alpha y_{2t}-z_{t}\beta -\gamma y_{1t-1} -u_{t} \} & , & t\in \mathcal{T}_{0} \\ 
\{(c,y_{10}) \in \mathbb{R}\times \mathcal{Y}_{10} : \alpha y_{2t} + z_{t}\beta +\gamma y_{1t-1} +c+u_{t}\leq y_{1t}\} & , & t\in \mathcal{T}_{1}
\end{array}%
\right. 
\end{equation*}
The $U^{\ast }$ sets are then as follows.
\begin{multline*}
\mathcal{U}^{\ast }(Y,Z;h)=\left\{ u: \exists Y_{10} \in \mathcal{Y}_{10} \text{ s.t. } \forall s,t \in \mathcal{T}_{0}\text{, } \Delta_{ts}u = \Delta_{ts}Y_1 -\alpha \Delta_{ts}Y_2 - \Delta_{ts}Z\beta - \gamma \Delta_{t-1,s-1}Y_1\right.  \\
\left. \wedge \text{ } \forall (s,t) \in \mathcal{T}_1\times \mathcal{T}_{0},\text{ } \Delta_{ts}u \geq \Delta_{ts} Y_{1} -\alpha \Delta_{ts}Y_2 -\Delta_{ts}Z\beta -\gamma\Delta_{t-1,s-1}Y_1 \right\}\text{,}
\end{multline*}
where $\Delta_{ts}U \equiv U_{t} - U_{s}$, $\Delta_{ts}Z \equiv Z_{t} - Z_{s}$, $\Delta_{ts}Y_1 \equiv Y_{1t} - Y_{1s}$ and $\Delta_{ts}Y_2 \equiv Y_{2t} - Y_{2s}$. 

Following the approach set out in Theorem \ref{Outerset th 1}, the identified set of values of $\left( \theta ,\mathcal{G}_{U|Z}\right) $, where $\theta\equiv\left(\alpha,\beta,\gamma \right)$ are those
satisfying the inequalities
\begin{equation*}
\mathbb{P}[Y\in \mathcal{Y}(\mathcal{T},z;\theta )|Z=z]\leq G_{U|Z=z}(\mathcal{S}%
(\mathcal{T},z;\theta ))\text{ a.e. } z \in \mathcal{R}_Z\text{,}
\end{equation*}%
for an appropriate selection of sets $\mathcal{T}$, where the sets $\mathcal{S}(\mathcal{T},z;\theta )$ and $\mathcal{Y}(\mathcal{T},z;\theta )$ are defined in   (\ref{Sfunction}) and (\ref{Yfunction}). The required selection can be characterized following the same steps taken in the models studied in prior sections. 

\section{Concluding remarks}

This paper delivers methods for producing identified sets when models admit
unobserved, latent, variables on which no distributional restrictions are
placed, opening the way to robust analysis of short panels. Examples found in econometric practice include models incorporating so-called fixed effects and initial conditions. Endogenous explanatory variables are easily
accommodated.

The identified sets delivered by the models in this paper that place no
restriction on the distribution of latent $V$ will contain the structures
identified by more restrictive models if the restrictions of those models
are satisfied by the process under study. The analysis set out here will show
how sensitive the findings obtained using that more restrictive model are to those additional restrictions. In some cases it may be found that estimation employing a point-identifying model delivers a structure outside an estimator of the identified set obtained using a less restrictive model of the type studied in this paper.
Such a finding would suggest the more restrictive model is misspecified.
Formal development of such specification tests may be of interest for future research.

\pagebreak 

\bibliographystyle{econometrica}
\bibliography{areference}

\pagebreak 
%\appendix
\begin{appendices}

\section{CR Restrictions A1-A6}

\label{Appendix: CR restrictions} This section collects restrictions from 
\cite{chesher2017generalized} adapted to the present setting with
unobservable variables $(U,V)$, which are imposed throughout the paper.

\noindent \textbf{Restriction A1}: $\left( Y,Z,U,V\right) $ are random
vectors defined on a probability space $\left( \Omega ,\mathsf{L},\mathbb{P}%
\right) $, endowed with the Borel sets on $\Omega $. The support of $\left(
Y,Z,U,V\right) $ is a subset of Euclidean space. $\square $

\noindent \textbf{Restriction A2}: A collection of conditional distributions%
\begin{equation*}
\mathcal{F}_{Y|Z}\equiv \left\{ F_{Y|Z}\left( \cdot |z\right) :z\in \mathcal{%
R}_{Z}\right\} \text{,}
\end{equation*}%
is identified by the sampling process, where for all $\mathcal{T}\subseteq 
\mathcal{R}_{Y|z}$, $F_{Y|Z}\left( \mathcal{T}|z\right) \equiv \mathbb{P}%
\left[ Y\in \mathcal{T}|z\right] $. $\square $

\noindent \textbf{Restriction A3:} There is an $\mathsf{L}$-measurable
function $h\left( \cdot ,\cdot ,\cdot, \cdot \right) :\mathcal{R}%
_{YZUV}\rightarrow \mathcal{%
%TCIMACRO{\U{211d} }%
%BeginExpansion
\mathbb{R}
%EndExpansion
}$ such that%
\begin{equation*}
\mathbb{P}\left[ h\left( Y,Z,U\right) =0\right] =1\text{,}
\end{equation*}%
and there is a collection of conditional distributions%
\begin{equation*}
\mathcal{G}_{U|Z}\equiv \left\{ G_{U|Z}\left( \cdot |z\right) :z\in \mathcal{%
R}_{Z}\right\} \text{,}
\end{equation*}%
where for all $\mathcal{S}\subseteq \mathcal{R}_{U|z}$, $G_{U|Z}\left( 
\mathcal{S}|z\right) \equiv \mathbb{P}\left[ U\in \mathcal{S}|z\right] $. $%
\square $

\noindent \textbf{Restriction A4}: The pair $\left( h,\mathcal{G}%
_{U|Z}\right) $ belongs to a known set of admissible structures $\mathcal{M}$%
. $\square $

\noindent \textbf{Restriction A5}: $\mathcal{U}^{\ast}\left( Y,Z;h\right) $ 
%\equiv \left\{ u \in \mathcal{R}_U:h\left( Y,Z,u\right)=0\right\}$
is closed almost surely $\mathbb{P}\left[ \cdot |z\right] $, each $z\in 
\mathcal{R}_{Z}$. $\square $

\noindent \textbf{Restriction A6}: $\mathcal{Y}^{\ast}\left(Z,U;h\right)$ 
%\equiv \left\{ y\in\mathcal{R}_Y:h\left( y,Z,U\right)=0\right\} \quad \text{and} \quad \mathcal{U}\left( Y,Z;h\right) \equiv \left\{ u \in \mathcal{R}_U:h\left( Y,Z,u\right)=0\right\}
%\end{equation*}
is closed almost surely $\mathbb{P}\left[ \cdot |z\right] $, each $z\in 
\mathcal{R}_{Z}$. $\square $\smallskip \pagebreak

\section{Sets $\mathcal{T}$ for sharp identified sets}

\label{Appendix: T tables} This section collects tables of $\mathcal{T}$ and 
$\mathcal{Y}(\mathcal{T},z;\theta )$ defined in (\ref{Yfunction}) as 
\begin{equation*}
\mathcal{Y}(\mathcal{T},z;\theta )\equiv \{y:\mathcal{U}^{\ast }(y,z;\theta
)\subseteq \mathcal{S}(\mathcal{T},z;\theta )\},
\end{equation*}
such that inequalities of the form 
\begin{equation*}
\mathbb{P}[Y\in \mathcal{Y}(\mathcal{T},z;\theta )|Z=z]\leq G_{U|Z=z}(%
\mathcal{S}(\mathcal{T},z;\theta ))
\end{equation*}
for all $\mathcal{T}$ listed characterize the identified set for $%
\left(\theta,\mathcal{G}_{U|Z} \right)$ in all examples covered in Sections %
\ref{Section: binary response}--\ref{Section: simultaneous binary response}.
Recall from (\ref{Sfunction}) the definition of $\mathcal{S}(\mathcal{T}%
,z;\theta )$:

\begin{equation*}
\mathcal{S}(\mathcal{T},z;\theta )\equiv \dbigcup\limits_{y\in \mathcal{T}}%
\mathcal{U}^{\ast }(y,z;\theta ) \text{.}
\end{equation*}

%TCIMACRO{\TeXButton{B}{\begin{table}[tbp] \centering}}%
%BeginExpansion
\begin{table}[tbp] \centering%
%EndExpansion
\caption{Sets $\mathcal{Y}(\mathcal{T},z;\theta )$ and $\mathcal{T}$ in the inequalities defining
the identified set of structures in the static binary response 3 period panel data model ($\gamma = 0$).}%
\medskip 
\begin{tabular}{|c|c|c|}
\hline
& $\mathcal{Y}(\mathcal{T},z;\theta )$ & $\mathcal{T}$ \\ \hline\hline
1 & $\left\{ (0,0,1)\right\} $ & $\left\{ (0,0,1)\right\} $ \\ \hline
2 & $\left\{ (0,1,0)\right\} $ & $\left\{ (0,1,0)\right\} $ \\ \hline
3 & $\left\{ (0,1,1)\right\} $ & $\left\{ (0,1,1)\right\} $ \\ \hline
4 & $\left\{ (1,0,0)\right\} $ & $\left\{ (1,0,0)\right\} $ \\ \hline
5 & $\left\{ (1,0,1)\right\} $ & $\left\{ (1,0,1)\right\} $ \\ \hline
6 & $\{(1,1,0)\}$ & $\{(1,1,0)\}$ \\ \hline
7 & $\{(0,0,1),(0,1,0),(0,1,1)\}$ & $\{(0,0,1),(0,1,0)\}$ \\ \hline
8 & $\{(0,0,1),(0,1,1)\}$ & $\{(0,0,1),(0,1,1)\}$ \\ \hline
9 & $\{(0,0,1),(1,0,0),(1,0,1)\}$ & $\{(0,0,1),(1,0,0)\}$ \\ \hline
10 & $\{(0,0,1),(1,0,1)\}$ & $\{(0,0,1),(1,0,1)\}$ \\ \hline
11 & $\{(0,1,0),(0,1,1)\}$ & $\{(0,1,0),(0,1,1)\}$ \\ \hline
12 & $\{(0,1,0),(1,0,0),(1,1,0)\}$ & $\{(0,1,0),(1,0,0)\}$ \\ \hline
13 & $\{(0,1,0),(1,1,0)\}$ & $\{(0,1,0),(1,1,0)\}$ \\ \hline
14 & $\{(0,0,1),(0,1,1),(1,0,1)\}$ & $\{(0,1,1),(1,0,1)\}$ \\ \hline
15 & $\{(0,1,0),(0,1,1),(1,1,0)\}$ & $\{(0,1,1),(1,1,0)\}$ \\ \hline
16 & $\{(1,0,0),(1,0,1)\}$ & $\{(1,0,0),(1,0,1)\}$ \\ \hline
17 & $\{(1,0,0),(1,1,0)\}$ & $\{(1,0,0),(1,1,0)\}$ \\ \hline
18 & $\{(1,0,0),(1,0,1),(1,1,0)\}$ & $\{(1,0,1),(1,1,0)\}$ \\ \hline
19 & $\{(0,0,1),(0,1,0),(0,1,1),(1,0,1)\}$ & $\{(0,0,1),(0,1,0),(1,0,1)\}$
\\ \hline
20 & $\{(0,0,1),(0,1,0),(0,1,1),(1,1,0)\}$ & $\{(0,0,1),(0,1,0),(1,1,0)\}$
\\ \hline
21 & $\{(0,0,1),(0,1,1),(1,0,0),(1,0,1)\}$ & $\{(0,0,1),(0,1,1),(1,0,0)\}$
\\ \hline
22 & $\{(0,0,1),(1,0,0),(1,0,1),(1,1,0)\}$ & $\{(0,0,1),(1,0,0),(1,1,0)\}$
\\ \hline
23 & $\{(0,1,0),(0,1,1),(1,0,0),(1,1,0)\}$ & $\{(0,1,0),(0,1,1),(1,0,0)\}$
\\ \hline
24 & $\{(0,1,0),(1,0,0),(1,0,1),(1,1,0)\}$ & $\{(0,1,0),(1,0,0),(1,0,1)\}$
\\ \hline
\end{tabular}%
\label{Table:SBP3 inequalities}%
%TCIMACRO{\TeXButton{E}{\end{table}}}%
%BeginExpansion
\end{table}%
%EndExpansion

%TCIMACRO{\TeXButton{B}{\begin{table}[tbp] \centering}}%
%BeginExpansion
\begin{table}[tbp] \centering%
%EndExpansion
\caption{Sets $\mathcal{Y}(\mathcal{T},z;\theta )$ and $\mathcal{T}$ in the
core determining  inequalities defining the identified set of structures in
the dynamic binary response 3 period panel with $Y_0$ not observed and
$\gamma > 0$.}\medskip 
\begin{tabular}{|c|c|c|}
\hline
& $\mathcal{Y}(\mathcal{T},z;\theta )$ & $\mathcal{T}$ \\ \hline\hline
1 & $\left\{ (0,0,1)\right\} $ & $\left\{ (0,0,1)\right\} $ \\ \hline
2 & $\left\{ (0,1,0)\right\} $ & $\left\{ (0,1,0)\right\} $ \\ \hline
3 & $\left\{ (0,1,1)\right\} $ & $\left\{ (0,1,1)\right\} $ \\ \hline
4 & $\left\{ (1,0,0)\right\} $ & $\left\{ (1,0,0)\right\} $ \\ \hline
5 & $\left\{ (1,0,1)\right\} $ & $\left\{ (1,0,1)\right\} $ \\ \hline
6 & $\{(1,1,0)\}$ & $\{(1,1,0)\}$ \\ \hline
7 & $\{(0,0,1),(0,1,1)\}$ & $\{(0,0,1),(0,1,1)\}$ \\ \hline
8 & $\{(0,0,1),(1,0,0),(1,0,1)\}$ & $\{(0,0,1),(1,0,0)\}$ \\ \hline
9 & $\{(0,0,1),(1,0,1)\}$ & $\{(0,0,1),(1,0,1)\}$ \\ \hline
10 & $\{(0,1,0),(0,1,1)\}$ & $\{(0,1,0),(0,1,1)\}$ \\ \hline
11 & $\{(0,1,0),(1,1,0)\}$ & $\{(0,1,0),(1,1,0)\}$ \\ \hline
12 & $\{(0,1,0),(0,1,1),(1,1,0)\}$ & $\{(0,1,1),(1,1,0)\}$ \\ \hline
13 & $\{(1,0,0),(1,0,1)\}$ & $\{(1,0,0),(1,0,1)\}$ \\ \hline
14 & $\{(1,0,0),(1,1,0)\}$ & $\{(1,0,0),(1,1,0)\}$ \\ \hline
15 & $\{(0,0,1),(0,1,0),(0,1,1)\}$ & $\{(0,0,1),(0,1,0),(0,1,1)\}$ \\ \hline
16 & $\{(0,0,1),(0,1,1),(1,0,0),(1,0,1)\}$ & $\{(0,0,1),(0,1,1),(1,0,0)\}$
\\ \hline
17 & $\{(0,0,1),(0,1,1),(1,0,1)\}$ & $\{(0,0,1),(0,1,1),(1,0,1)\}$ \\ \hline
18 & $\{(0,0,1),(0,1,0),(0,1,1),(1,1,0)\}$ & $\{(0,0,1),(0,1,1),(1,1,0)\}$
\\ \hline
19 & $\{(0,0,1),(1,0,0),(1,0,1),(1,1,0)\}$ & $\{(0,0,1),(1,0,0),(1,1,0)\}$
\\ \hline
20 & $\{(0,1,0),(0,1,1),(1,0,0),(1,1,0)\}$ & $\{(0,1,0),(0,1,1),(1,0,0)\}$
\\ \hline
21 & $\{(0,1,0),(1,0,0),(1,1,0)\}$ & $\{(0,1,0),(1,0,0),(1,1,0)\}$ \\ \hline
22 & $\{(1,0,0),(1,0,1),(1,1,0)\}$ & $\{(1,0,0),(1,0,1),(1,1,0)\}$ \\ \hline
23 & $\{(0,0,1),(0,1,0),(0,1,1),(1,0,1)\}$ & $%
\{(0,0,1),(0,1,0),(0,1,1),(1,0,1)\}$ \\ \hline
24 & $\{(0,0,1),(0,1,0),(1,0,0),(1,0,1)(1,1,0)\}$ & $%
\{(0,0,1),(0,1,0),(1,0,0),(1,1,0)\}$ \\ \hline
25 & $\{(0,0,1),(0,1,0),(0,1,1),(1,0,1),(1,1,0)\}$ & $%
\{(0,0,1),(0,1,1),(1,0,1),(1,1,0)\}$ \\ \hline
26 & $\{(0,1,0),(1,0,0),(1,0,1),(1,1,0)\}$ & $%
\{(0,1,0),(1,0,0),(1,0,1),(1,1,0)\}$ \\ \hline
\end{tabular}%
\label{Table: DBP3 core determining inequalities gamma positive}%
%TCIMACRO{\TeXButton{E}{\end{table}}}%
%BeginExpansion
\end{table}%
%EndExpansion

%TCIMACRO{\TeXButton{B}{\begin{table}[tbp] \centering}}%
%BeginExpansion
\begin{table}[tbp] \centering%
%EndExpansion
\caption{Sets $\mathcal{Y}(\mathcal{T},z;\theta )$ and $\mathcal{T}$ in the inequalities defining
the identified set of structures in the dynamic binary response 3 period panel with $Y_0$ not observed and $\gamma < 0$.}%
\medskip 
\begin{tabular}{|c|c|c|}
\hline
& $\mathcal{Y}(\mathcal{T},z;\theta )$ & $\mathcal{T}$ \\ \hline\hline
1 & $\left\{ (0,0,1)\right\} $ & $\left\{ (0,0,1)\right\} $ \\ \hline
2 & $\left\{ (0,1,0)\right\} $ & $\left\{ (0,1,0)\right\} $ \\ \hline
3 & $\left\{ (0,1,1)\right\} $ & $\left\{ (0,1,1)\right\} $ \\ \hline
4 & $\left\{ (1,0,0)\right\} $ & $\left\{ (1,0,0)\right\} $ \\ \hline
5 & $\left\{ (1,0,1)\right\} $ & $\left\{ (1,0,1)\right\} $ \\ \hline
6 & $\{(1,1,0)\}$ & $\{(1,1,0)\}$ \\ \hline
7 & $\{(0,0,1),(0,1,0),(0,1,1\}$ & $\{(0,0,1),(0,1,0)\}$ \\ \hline
8 & $\{(0,0,1),(0,1,1)\}$ & $\{(0,0,1),(0,1,1)\}$ \\ \hline
9 & $\{(0,0,1),(1,0,0)\}$ & $\{(0,0,1),(1,0,0)\}$ \\ \hline
10 & $\{(0,0,1),(1,0,1)\}$ & $\{(0,0,1),(1,0,1)\}$ \\ \hline
11 & $\{(0,1,0),(0,1,1)\}$ & $\{(0,1,0),(0,1,1)\}$ \\ \hline
12 & $\{(0,1,0),(1,0,0),(1,1,0)\}$ & $\{(0,1,0),(1,0,0)\}$ \\ \hline
13 & $\{(0,1,0),(1,0,1)\}$ & $\{(0,1,0),(1,0,1)\}$ \\ \hline
14 & $\{(0,1,0),(1,1,0)\}$ & $\{(0,1,0),(1,1,0)\}$ \\ \hline
15 & $\{(0,0,1),(0,1,1),(1,0,1)\}$ & $\{(0,1,1),(1,0,1)\}$ \\ \hline
16 & $\{(0,1,1),(1,1,0)\}$ & $\{(0,1,1),(1,1,0)\}$ \\ \hline
17 & $\{(1,0,0),(1,0,1)\}$ & $\{(1,0,0),(1,0,1)\}$ \\ \hline
18 & $\{(1,0,0),(1,1,0)\}$ & $\{(1,0,0),(1,1,0)\}$ \\ \hline
19 & $\{(1,0,0),(1,0,1),(1,1,0)\}$ & $\{(1,0,1),(1,1,0)\}$ \\ \hline
20 & $\{(0,0,1),(0,1,0),(0,1,1),(1,0,1)\}$ & $\{(0,0,1),(0,1,0),(1,0,1)\}$
\\ \hline
21 & $\{(0,0,1),(0,1,0),(0,1,1),(1,1,0)\}$ & $\{(0,0,1),(0,1,0),(1,1,0)\}$
\\ \hline
22 & $\{(0,0,1),(0,1,1),(1,0,0),(1,0,1)\}$ & $\{(0,0,1),(0,1,1),(1,0,0)\}$
\\ \hline
23 & $\{(0,0,1),(1,0,0),(1,0,1)\}$ & $\{(0,0,1),(1,0,0),(1,0,1)\}$ \\ \hline
24 & $\{(0,0,1),(1,0,0),(1,0,1),(1,1,0)$ & $\{(0,0,1),(1,0,0),(1,1,0)\}$ \\ 
\hline
25 & $\{(0,1,0),(0,1,1),(1,0,0),(1,1,0)\}$ & $\{(0,1,0),(0,1,1),(1,0,0)\}$
\\ \hline
26 & $\{(0,1,0),(0,1,1),(1,1,0)\}$ & $\{(0,1,0),(0,1,1),(1,1,0)\}$ \\ \hline
27 & $\{(0,1,0),(1,0,0),(1,0,1),(1,1,0)\}$ & $\{(0,1,0),(1,0,0),(1,0,1)\}$
\\ \hline
\end{tabular}%
\label{Table: DBP3 inequalities gamma negative}%
%TCIMACRO{\TeXButton{E}{\end{table}}}%
%BeginExpansion
\end{table}%
%EndExpansion

%TCIMACRO{\TeXButton{B}{\begin{table}[tbp] \centering}}%
%BeginExpansion
\begin{table}[tbp] \centering%
%EndExpansion
\caption{Sets $\mathcal{Y}(\mathcal{T},z;\theta )$ and $\mathcal{T}$ in the inequalities defining
the identified set of structures in the 3 choice multiple discrete choice 2 period panel data model.}%
\medskip 
\begin{tabular}{|c|c|c|}
\hline
& $\mathcal{Y}(\mathcal{T},z;\theta )$ & $\mathcal{T}$ \\ \hline\hline
1 & $\left\{ (1,2)\right\} $ & $\left\{ (1,2)\right\} $ \\ \hline
2 & $\left\{ (1,3)\right\} $ & $\left\{ (1,3)\right\} $ \\ \hline
3 & $\left\{ (2,1)\right\} $ & $\left\{ (2,1)\right\} $ \\ \hline
4 & $\left\{ (2,3)\right\} $ & $\left\{ (2,3)\right\} $ \\ \hline
5 & $\left\{ (3,1)\right\} $ & $\left\{ (3,1)\right\} $ \\ \hline
6 & $\left\{ (3,2)\right\} $ & $\left\{ (3,2)\right\} $ \\ \hline\hline
7 & $\{(1,2),(1,3)\}$ & $\{(1,2),(1,3)\}$ \\ \hline
8 & $\{(1,2),(3,2)\}$ & $\{(1,2),(3,2)\}$ \\ \hline
9 & $\{(1,3),(2,3)\}$ & $\{(1,3),(2,3)\}$ \\ \hline
10 & $\{(2,1),(2,3)\}$ & $\{(2,1),(2,3)\}$ \\ \hline
11 & $\{(2,1),(3,1)\}$ & $\{(2,1),(3,1)\}$ \\ \hline
12 & $\{(3,1),(3,2)\}$ & $\{(3,1),(3,2)\}$ \\ \hline\hline
13 & $\{(1,2),(1,3),(2,3)\}$ & $\{(1,2),(2,3)\}$ \\ \hline
14 & $\{(1,2),(1,3),(3,2)\}$ & $\{(1,3),(3,2)\}$ \\ \hline
15 & $\{(1,2),(3,1),(3,2)\}$ & $\{(1,2),(3,1)\}$ \\ \hline
16 & $\{(1,3),(2,1),(2,3)\}$ & $\{(1,3),(2,1)\}$ \\ \hline
17 & $\{(2,1),(2,3),(3,1)\}$ & $\{(2,3),(3,1)\}$ \\ \hline
18 & $\{(2,1),(3,1),(3,2)\}$ & $\{(2,1),(3,2)\}$ \\ \hline
\end{tabular}%
\label{Table: MDC inequalities using T sets}%
%TCIMACRO{\TeXButton{E}{\end{table}}}%
%BeginExpansion
\end{table}%
%EndExpansion

%TCIMACRO{\TeXButton{B}{\begin{table}[tbp] \centering}}%
%BeginExpansion
\begin{table}[tbp] \centering%
%EndExpansion
\caption{Sets $\mathcal{Y}(\mathcal{T},z;\theta )$ and $\mathcal{T}$ in the inequalities defining
the identified set of values of $\theta$  in the simultaneous binary response 2 period panel.}
\medskip 
\begin{tabular}{|c|c|c|}
\hline
& $\mathcal{Y}(\mathcal{T},z;\theta )$ & $\mathcal{T}$ \\ \hline
1 & $\left\{ (0,0,0,1),(1,0,0,1),(1,1,0,1)\right\} $ & $\left\{
(0,0,0,1)\right\} $ \\ \hline
2 & $\left\{ (0,0,1,0),(0,1,1,0),(1,1,1,0)\right\} $ & $\left\{
(0,0,1,0)\right\} $ \\ \hline
3 & $\left\{ (0,1,0,0),(0,1,1,0),(0,1,1,1)\right\} $ & $\left\{
(0,1,0,0)\right\} $ \\ \hline
4 & $\left\{ (1,0,0,0),(1,0,0,1),(1,0,1,1)\right\} $ & $\left\{
(1,0,0,0)\right\} $ \\ \hline
5 & $\left\{ (0,1,0,1)\right\} $ & $\left\{ (0,1,0,1)\right\} $ \\ \hline
6 & $\left\{ (0,1,1,0)\right\} $ & $\left\{ (0,1,1,0)\right\} $ \\ \hline
7 & $\left\{ (1,0,0,1)\right\} $ & $\left\{ (1,0,0,1)\right\} $ \\ \hline
8 & $\left\{ (1,0,1,0)\right\} $ & $\left\{ (1,0,1,0)\right\} $ \\ \hline
9 & $\left\{ 
\begin{array}{c}
(0,0,0,1),(0,1,0,0),(0,1,1,0), \\ 
(0,1,1,1),(1,0,0,1),(1,1,0,1)%
\end{array}%
\right\} $ & $\{(0,0,0,1),(0,1,0,0)\}$ \\ \hline
10 & $\left\{ 
\begin{array}{c}
(0,0,0,1),(1,0,0,0),(1,0,0,1), \\ 
(1,0,1,1),(1,1,0,1)%
\end{array}%
\right\} $ & $\{(0,0,0,1),(1,0,0,0)\}$ \\ \hline
11 & $\{(0,0,0,1),(0,1,0,1),(1,0,0,1),(1,1,0,1)\}$ & $\{(0,0,0,1),(0,1,0,1)%
\} $ \\ \hline
12 & $\left\{ 
\begin{array}{c}
(0,0,0,1),(1,0,0,0),(1,0,0,1), \\ 
(1,0,1,0),(1,0,1,1),(1,1,0,1)%
\end{array}%
\right\} $ & $\{(0,0,0,1),(1,0,1,0)\}$ \\ \hline
13 & $\left\{ 
\begin{array}{c}
(0,0,1,0),(0,1,0,0),(0,1,1,0), \\ 
(0,1,1,1),(1,1,1,0)%
\end{array}%
\right\} $ & $\{(0,0,1,0),(0,1,0,0)\}$ \\ \hline
14 & $\left\{ 
\begin{array}{c}
(0,0,1,0),(0,1,1,0),(1,0,0,0), \\ 
(1,0,0,1),(1,0,1,1),(1,1,1,0)%
\end{array}%
\right\} $ & $\{(0,0,1,0),(1,0,0,0)\}$ \\ \hline
15 & $\left\{ 
\begin{array}{c}
(0,0,1,0),(0,1,0,0),(0,1,0,1), \\ 
(0,1,1,0),(0,1,1,1),(1,1,1,0)%
\end{array}%
\right\} $ & $\{(0,0,1,0),(0,1,0,1)\}$ \\ \hline
16 & $\{(0,0,1,0),(0,1,1,0),(1,0,1,0),(1,1,1,0)\}$ & $\{(0,0,1,0),(1,0,1,0)%
\} $ \\ \hline
17 & $\{(0,1,0,0),(0,1,0,1),(0,1,1,0),(0,1,1,1)\}$ & $\{(0,1,0,0),(0,1,0,1)%
\} $ \\ \hline
18 & $\left\{ 
\begin{array}{c}
(0,0,1,0),(0,1,0,0),(0,1,1,0), \\ 
(0,1,1,1),(1,0,1,0),(1,1,1,0)%
\end{array}%
\right\} $ & $\{(0,1,0,0),(1,0,1,0)\}$ \\ \hline
19 & $\left\{ 
\begin{array}{c}
(0,0,0,1),(0,1,0,1),(1,0,0,0), \\ 
(1,0,0,1),(1,0,1,1),(1,1,0,1)%
\end{array}%
\right\} $ & $\{(1,0,0,0),(0,1,0,1)\}$ \\ \hline
20 & $\{(1,0,0,0),(1,0,0,1),(1,0,1,0),(1,0,1,1)\}$ & $\{(1,0,0,0),(1,0,1,0)%
\} $ \\ \hline
21 & $\{(0,1,0,1),(1,0,1,0)\}$ & $\{(0,1,0,1),(1,0,1,0)\}$ \\ \hline
22 & $\left\{ 
\begin{array}{c}
(0,0,0,1),(0,1,0,0),(0,1,0,1),(0,1,1,0), \\ 
(0,1,1,1),(1,0,0,1),(1,1,0,1)%
\end{array}%
\right\} $ & $\{(0,0,0,1),(0,1,0,0),(0,1,0,1)\}$ \\ \hline
23 & $\left\{ 
\begin{array}{c}
(0,0,0,1),(0,1,0,1),(1,0,0,0),(1,0,0,1), \\ 
(1,0,1,0),(1,0,1,1),(1,1,0,1)%
\end{array}%
\right\} $ & $\{(0,0,0,1),(0,1,0,1),(1,0,1,0)\}$ \\ \hline
24 & $\left\{ 
\begin{array}{c}
(0,0,1,0),(0,1,1,0),(1,0,0,0),(1,0,0,1), \\ 
(1,0,1,0),(1,0,1,1),(1,1,1,0)%
\end{array}%
\right\} $ & $\{(0,0,1,0),(1,0,0,0),(1,0,1,0)\}$ \\ \hline
25 & $\left\{ 
\begin{array}{c}
(0,0,1,0),(0,1,0,0),(0,1,0,1),(0,1,1,0), \\ 
(0,1,1,1),(1,0,1,0),(1,1,1,0)%
\end{array}%
\right\} $ & $\{(0,0,1,0),(0,1,0,1),(1,0,1,0)\}$ \\ \hline
\end{tabular}%
\label{Table:SES inequalities}%
%TCIMACRO{\TeXButton{E}{\end{table}}}%
%BeginExpansion
\end{table}%
%EndExpansion
\pagebreak

\section{Bounds on structural parameters $\theta$}\label{Appendix: Profile Bounds}

Here calculation of the projection of the identified set of values
of $(\theta ,G_{U|Y_{0}})$ onto the space of $\theta $ as set out in Remark \ref{NumberedList: capcon} of section \ref{Section: Remarks} is demonstrated for the two-period dynamic binary reponse model studied in Section \ref{Section: 2 period dynamic panel}.

Define 
\begin{equation*}
p_{jk}\left(z,y_0\right)\equiv \mathbb{P}\left[Y=(j,k)|Z=z,Y_0=y_0\right]%
\text{.}
\end{equation*}
Consider sets $\left\{u:\Delta u \in (-\infty ,w] \right\}$ and $%
\left\{u:\Delta u \in (w,\infty) \right\}$ for $w\in \mathbb{R}$. There is 
\begin{multline*}
\mathbb{P}[\mathcal{U}^{\ast }(Y,Z;\theta )\subseteq \left\{u:\Delta u \in
(-\infty ,w] \right\}|Z=z,Y_{0}=y_{0}]= \\
p_{10}\left(z,y_0\right)\times 1[-\Delta z\beta +(y_{0}-1)\gamma \leq w]%
\text{,}
\end{multline*}
and 
\begin{multline*}
\mathbb{P}[\mathcal{U}^{\ast }(Y,Z;\theta )\subseteq \left\{u:\Delta u \in
(w,\infty) \right\} |Z=z,Y_{0}=y_{0}]= \\
p_{01}\left(z,y_0\right) \times 1[-\Delta z\beta +y_{0}\gamma > w]\text{.}
\end{multline*}
Thus, using $\mathcal{U}^{\ast}(Y,Z;\theta)\subseteq \mathcal{S} \implies U
\in \mathcal{S}$ we have the inequalities 
\begin{equation*}
p_{10}\left(z,y_0\right)\times 1[-\Delta z\beta +(y_{0}-1)\gamma \leq w]
\leq G_{U|Y_{0}=y_{0}}\left( \left\{ u:\Delta u \leq w \right\} \right)
\end{equation*}
and 
\begin{equation*}
p_{01}\left(z,y_0\right) \times 1[-\Delta z\beta +y_{0}\gamma > w] \leq
G_{U|Y_{0}=y_{0}}\left( \left\{ u:\Delta u > w \right\} \right)\text{.}
\end{equation*}
Let $\tilde{G}_{\Delta U}(w;y_0)\equiv G_{U|Y_{0}=y_{0}}\left( \left\{
u:\Delta u \leq w \right\}\right)$. Since $G_{U|Y_{0}=y_{0}}\left( \left\{
u:\Delta u > w \right\} \right) = 1 - G_{U|Y_{0}=y_{0}}\left( \left\{
u:\Delta u \leq w \right\} \right)$ it follows that 
\begin{equation}  \label{profile inequality binary}
p_{10}\left(z,y_0\right)\times 1[-\Delta z\beta +(y_{0}-1)\gamma \leq w]
\leq \tilde{G}_{\Delta U}(w;y_0) \leq 1 - p_{01}\left(z,y_0\right) \times
1[-\Delta z\beta +y_{0}\gamma > w]\text{.}
\end{equation}
Bounds on $\theta$ are obtained as those values for which the above lower
and upper inequalities never cross, as follows: 
\begin{equation}
\Theta^{\ast} \equiv \left\{\theta \in \Theta: \forall (y_0,w) \in
\{0,1\}\times \mathbb{R} \quad \tilde{G}^L_{\Delta U}(w;y_0) \leq \tilde{G}%
^U_{\Delta U}(w;y_0) \right\}\text{,}
\end{equation}
where $\tilde{G}^L_{\Delta U}(w;y_0)$ and $\tilde{G}^U_{\Delta U}(w;y_0)$
correspond to the lower and upper envelopes for $\tilde{G}_{\Delta U}(w;y_0)$
obtained from (\ref{profile inequality binary}) upon taking intersections
across $z$: 
\begin{equation*}
\tilde{G}^L_{\Delta U}(w;y_0) \equiv \left\{ 
\begin{array}{cc}
\sup\limits_{z \in \mathcal{Z}_L(w,y_0)} p_{10}\left(z,y_0\right) & \text{if}
\quad \mathcal{Z}_L(w,y_0) \neq \emptyset \text{,} \\ 
0 & \text{otherwise.}%
\end{array}
\right.
\end{equation*}
\begin{equation*}
\tilde{G}^U_{\Delta U}(w;y_0) \equiv \left\{ 
\begin{array}{cc}
\inf\limits_{z \in \mathcal{Z}_U(w,y_0)} 1-p_{01}\left(z,y_0\right) & \text{%
if} \quad \mathcal{Z}_U(w,y_0) \neq \emptyset \text{,} \\ 
1 & \text{otherwise.}%
\end{array}
\right.
\end{equation*}
where 
\begin{align*}
\mathcal{Z}_L(w,y_0) &\equiv \left\{z\in\mathcal{R}_Z:-\Delta z\beta
+(y_{0}-1)\gamma \leq w \right\}\text{,} \\
\mathcal{Z}_U(w,y_0) &\equiv \left\{z\in\mathcal{R}_Z:-\Delta z\beta
+y_{0}\gamma > w \right\}{\text.}
\end{align*}
The bounds $\Theta^{\ast}$ correspond to those of Theorem 2 of \cite%
{aristodemou2021semiparametric}, up to minor notational differences.

In a model in which $U$ and $Z$ are independent conditional on $Y_0$ with no
restrictions placed on $G_{U|Y_0 = y_0}$, the set $\Theta^{\ast}$ thus
obtained is in fact the sharp identified set for $\theta $ because for every
value of $\tilde{\theta} \in \Theta^{\ast}$ there is for each $%
y_{0}\in\{0,1\}$ a distribution of $\Delta U$ conditional on $Y_{0}=y_{0}$,
say $\tilde{G}_{\Delta U|Y_{0}=y_0}$, such that $(\tilde{\theta},\tilde{G}%
_{\Delta U|Y_{0}=0},\tilde{G}_{\Delta U|Y_{0}=1})$ is contained in the
identified set of structures.

This is so because $\tilde{G}^L_{\Delta U}(w;y_0)$ and $\tilde{G}^U_{\Delta
U}(w;y_0)$ are nondecreasing functions of $w$ taking values on the unit
interval. Accordingly, for each $\tilde{\theta} \in \Theta^{\ast}$ there
exists for each value of $y_{0}$ a proper distribution $\tilde{G}_{\Delta
U|Y_{0}=y_0}(\cdot ,\tilde{\theta} )$ such that for all $w\in 
%TCIMACRO{\U{211d} }%
%BeginExpansion
\mathbb{R}
%EndExpansion
$,%
\begin{equation*}
\tilde{G}^L_{\Delta U}(w;y_0) \leq \tilde{G}_{\Delta U|Y_{0}=y_{0}}((-\infty
,w],\tilde\theta )\leq \tilde{G}^U_{\Delta U}(w;y_0)\text{,}
\end{equation*}
for example 
\begin{equation*}
\tilde{G}_{\Delta U|Y_{0}=y_{0}}((-\infty ,w],\theta )=\lambda \tilde{G}%
^L_{\Delta U}(w;y_0) +(1-\lambda )\tilde{G}^U_{\Delta U}(w;y_0)
\end{equation*}%
for any $\lambda \in (0,1)$.\footnote{%
This sharpness result relies on the identified set of structures being
determined by the distribution of a scalar function of the unobserved
variables, in this case $\Delta U$. When $T>2$ this will not be the case in
this example but it is the case in the model of section \ref{Section: 3
period dynamic y0 unobserved}.} Thus any distributions $\tilde{G}_{U|Y_0 =
y_0}(\cdot)$ for $U$ with corresponding distributions $\tilde{G}_{\Delta
U|Y_{0}=y_{0}}(\cdot,\tilde\theta )$ for each $y_0$ paired with $\tilde\theta
$ can produce the observed distributions of $Y$ given $(Y_0,Z)$.
\end{appendices}

\end{document}